\documentclass[pra,reprint,twocolumn,superscriptaddress,longbibliography]{revtex4-2}

\usepackage{amsmath,amssymb,amsthm}
\usepackage{graphicx}

\usepackage{xr-hyper}
\usepackage[colorlinks=true,linkcolor=teal,citecolor=teal,urlcolor=teal,plainpages=false,pdfpagelabels]{hyperref}
\usepackage{color}
\usepackage[dvipsnames]{xcolor}
\usepackage{mathtools}

\usepackage{float}
\usepackage{array}
\usepackage{verbatim}

\usepackage[utf8]{inputenc}
\usepackage[T1]{fontenc}
\usepackage{etoolbox}
\usepackage{titletoc} 

\usepackage{anyfontsize}  
\usepackage{microtype} 

\usepackage{newtxtext}
\usepackage{newtxmath}

\usepackage{dsfont}
\usepackage{soul}
\usepackage{easy-todo}
\usepackage{layouts} 

\newcommand{\<}{\langle}

\newcommand{\II}{\mathds{1}}

\DeclareMathOperator{\Tr}{Tr}
\newcommand{\ket}[1]{|#1\rangle}
\newcommand{\bra}[1]{\langle#1|}
\newcommand{\ketbra}[2]{\ket{#1}\!\bra{#2}}
\newcommand{\braket}[2]{\langle #1|#2\rangle}

\makeatletter
\g@addto@macro\bfseries{\boldmath}
\makeatother

\theoremstyle{definition}

\newtheorem{definition}{Definition}

\newtheorem{theorem}{Theorem}
\newtheorem{corollary}{Corollary}
\newtheorem{lemma}{Lemma}

\newcommand{\proj}[1]{\ket{#1}\!\bra{#1}}
\newcommand{\abs}[1]{\lvert #1 \rvert}
\newcommand{\norm}[1]{\lVert #1 \rVert}

\begin{document}
\title{Perturbative gadgets for gate-based quantum computing: \\ Non-recursive constructions without subspace restrictions}

\author{\href{https://orcid.org/0000-0002-9409-193X}{Simon~Cichy}}
\affiliation{Dahlem Center for Complex Quantum Systems, Freie Universit\"{a}t Berlin, 14195 Berlin, Germany}
\affiliation{Institute for Theoretical Physics, ETH Z\"{u}rich, 8093 Z\"{u}rich, Switzerland}
	
\author{\href{https://orcid.org/0000-0002-8706-1732}{Paul~K.~Faehrmann}}
\thanks{SC and PKF have contributed equally.}
\affiliation{Dahlem Center for Complex Quantum Systems, Freie Universit\"{a}t Berlin, 14195 Berlin, Germany}

\author{\href{https://orcid.org/0000-0002-9858-0511}{Sumeet~Khatri}}
\affiliation{Dahlem Center for Complex Quantum Systems, Freie Universit\"{a}t Berlin, 14195 Berlin, Germany}
	
\author{\href{https://orcid.org/0000-0003-3033-1292}{Jens~Eisert}}
\affiliation{Dahlem Center for Complex Quantum Systems, Freie Universit\"{a}t Berlin, 14195 Berlin, Germany}
\affiliation{Helmholtz-Zentrum Berlin f\"{u}r Materialien und Energie, Hahn-Meitner-Platz 1, 14109 Berlin, Germany}
\affiliation{Fraunhofer Heinrich Hertz Institute, 10587 Berlin, Germany}
	
\date{\today}

\begin{abstract}
    Perturbative gadgets are a tool to encode part of a Hamiltonian, usually the low-energy subspace, into a different Hamiltonian with favorable properties, for instance, reduced locality. Many constructions of perturbative gadgets have been proposed over the years. 
    Still, all of them are restricted in some ways: Either they apply to some specific classes of Hamiltonians, they involve recursion to reduce locality, or they are limited to studying time evolution under the gadget Hamiltonian, e.g., in the context of adiabatic quantum computing, and thus involve subspace restrictions. 
    In this work, we fill the gap by introducing a versatile universal, non-recursive, non-adiabatic perturbative gadget construction without subspace restrictions, that encodes an arbitrary many-body Hamiltonian into the low-energy subspace of a three-body Hamiltonian and is therefore applicable to gate-based quantum computing.
    Our construction requires $rk$ additional qubits for a $k$-body Hamiltonian comprising $r$ terms.
    Besides a specific gadget construction, we also provide a recipe for constructing similar gadgets, which can be tailored to different properties, which we discuss.
\end{abstract}

\maketitle

\startcontents[mainsections] 

\section{Introduction}

The study of many-body Hamiltonians is a well-researched field in condensed-matter physics and quantum information theory. 
While the presence of Hamiltonian terms acting on many qubits simultaneously can result in exciting phenomena, in many situations, it also carries additional difficulties.
Be it due to the hardness of generating them experimentally or other limitations, local Hamiltonians are usually simpler to deal with.

Born in the context of complexity theory for proving the QMA-completeness of the local Hamiltonian problem~\cite{kitaevClassicalQuantumComputation2002, kempe3-LocalHamiltonianIsQMA2003},
a problem located at the interface of
the theory of quantum many-body physics and Hamiltonian complexity,
so-called \emph{perturbative gadgets} can be used to reduce the locality of a given many-body Hamiltonian.
This is done by embedding it in the low-energy subspace of a tailored, local, i.e., few-body, Hamiltonian acting on a larger Hilbert space~\cite{kempeComplexityLocalHamiltonian2006,oliveiraComplexityQuantumSpin2008a}. 
Among the best-known constructions is the subdivision gadget proposed by Oliveira and Terhal~\cite{oliveiraComplexityQuantumSpin2008a}.
It employs mediator qubits to halve the support of a given Hamiltonian term and can be recursively applied to construct an at best three-body Hamiltonian mimicking the original $k$-body Hamiltonian.
As is typical for perturbative gadgets, such a procedure, unfortunately, leads to an exponential increase of the interaction strengths required for the gadget Hamiltonian to accurately mimic the low-energy subspace of the original Hamiltonian.

A different gadget, proposed by Jordan and Farhi~\cite{jordanPerturbativeGadgetsArbitrary2008a}, can avoid the need for recursion and offers a direct reduction from $k$-body to two-body Hamiltonians, albeit with an exponential suppression of energies.
However, their direct construction is only applicable when one can restrict the evolution of chosen initial states to predefined subspaces of the Hilbert space, as can be done in adiabatic quantum computing, where commutation relations guarantee the system remains in the initialized subspace.

Since such a limitation restricts the applicability of this direct $k$-to-two-local gadget, the question arises whether there exists a direct, non-recursive $k$-to-$k'$-local gadget without such restrictions and whether such a gadget could have any applications besides being a tool for proving complexity theoretic results.
After all, perturbative gadgets have so far mostly found their application in proving hardness results for the general local Hamiltonian problem~\cite{kempeComplexityLocalHamiltonian2006,oliveiraComplexityQuantumSpin2008a} or certain classes of Hamiltonians~\cite{cubittUniversalHamiltonians2018,bravyiComplexityIsing2017,zhouStronglyUniversalHamiltonian2021,piddockUniversalQuditHamiltonians2018,piddockUniversalTranslationallyInvariantHamiltonians2020} equipped with further constraints and properties.

In this work, we address these questions by proposing a direct $k$-to-three-local gadget without the need for subspace restrictions, that is inspired by Ref.~\cite{jordanPerturbativeGadgetsArbitrary2008a} but that avoids initialization and commutation properties from adiabatic quantum computing to achieve the same locality reduction as the repeated subdivision procedure~\cite{oliveiraComplexityQuantumSpin2008a}, without requiring recursive applications.

As expected of gadget constructions~\cite{harleyGoingGadgetsImportance2023}, our construction is also not free from unfavorable energy scalings, but we still
explore the use of such a gadget within gate-based quantum computing and, in particular, as a method for reducing the locality of the cost function in variational quantum algorithms, a setting that does not comply with the ``adiabatic'' restriction, or the grouping of measurement terms.

While most likely unfavorable for situations where the locality of a Hamiltonian scales with the system size, there are plenty of situations where the locality is constant but still high enough to cause problems. 
Examples include measurement noise scaling with the locality of the measurement operator influencing noise-induced barren plateaus~\cite{wangNoiseInducedBarrenPlateaus2021} or time evolution and quantum phase estimation for restricted hardware settings.
The proposed methods do not rely on subspace restrictions and can thus be implemented in gate-based quantum computing.
By providing a more general recipe for constructing perturbative gadgets with similar performance guarantees and suggesting plausible heuristics, we hope to start a discussion about how the trade-offs of perturbative gadgets could still be used for practical problems.

We begin by introducing the main ideas and a historical overview of perturbative gadgets in Section~\ref{sec:gadget-overview}, before focusing on a specific $k$-to-three-local gadget and discussing guarantees on its performance in Section~\ref{sec:new-gadget}.
We then generalize the gadget construction to provide a recipe for creating custom perturbative gadgets in Section~\ref{sec:recipe} and discuss the nuances of applying such gadgets in gate-based quantum computing in Section~\ref{sec:practical-applications}, before concluding with some remarks and open questions in Section~\ref{sec:outlook}.

\section{Perturbative gadgets in a nutshell}
\label{sec:gadget-overview}

\begin{table*}
\begin{center}
\renewcommand*{\arraystretch}{1.4}
\begin{tabular}{>{\raggedright}p{0.19\textwidth}
                >{\raggedright\arraybackslash}p{0.70\textwidth}
                }
    Authors & 
    Main statement 
    \\ [3pt]
    \hline
    Kempe, Kitaev, Regev (KKR)~\cite{kempeComplexityLocalHamiltonian2006} & 
    3- to 2-local gadget to prove QMA completeness of the local Hamiltonian problem.
    Requires restriction of the state evolution to a subspace of the Hilbert space.
    \\
    Oliveira, Terhal (OT)~\cite{oliveiraComplexityQuantumSpin2008a} & 
    Subdivision gadget ($k$- to $k/2$-local) and proof of QMA completeness of the Hamiltonian problem on a square grid. Has to be applied recursively resulting in exponential coupling strengths. 
    \\
    Biamonte, Love~\cite{biamonteRealizableHamiltoniansUniversal2008} & 
    ``Realizable Hamiltonians'' with a focus on using gadgets for adiabatic computing and adapting to hardware restrictions, not targeting a reduction of locality.
    \\
    Bravyi, DiVicenzo, Loss, Terhal~\cite{bravyiQuantumSimulationManyBody2008} & 
    Improvements on the OT gadget through extension to non-converging perturbative expansion, thus reducing the gap in coupling strengths 
    \\
    Jordan, Farhi~\cite{jordanPerturbativeGadgetsArbitrary2008a} & 
    Direct $k$- to $2$-local gadget through higher-order perturbation theory. 
    Requires restriction of the state evolution to a subspace of the Hilbert space.
    \\
    Cao, Babbush, Biamonte, Kais~\cite{Cao2015} & 
    Resource requirement improvements of some of the OT and KKR previous gadgets. 
    \\
    Cao~\cite{caoCombinatorialAlgorithmsPerturbation2016} & 
    Overview of existing gadgets and improvements of the OT and KKR gadgets. 
    \\
    Cao, Kais~\cite{caoEfficientOptimizationPerturbative2017} & 
    Improvement of the convergence bound for the JF gadget. 
    \\
    Subasi, Jarzynski~\cite{subasiNonperturbativeEmbeddingHighly2016} & 
    Nonperturbative gadgets for adiabatic Hamiltonian evolution. 
    \\
    Bausch~\cite{bauschGadgets2020} & 
    Augmentation of other gadget constructions and changes their energy scales by (potentially) increasing their locality by one.
    \\
    Harley, Datta, Klausen, Bluhm, França, Werner, Christandl~\cite{harleyGoingGadgetsImportance2023} & 
    A general framework for analog quantum simulation and unavoidable unfavorable size-dependent scalings of gadget constructions. \\
    This work & 
    Direct $k$- to $3$-local gadget without restriction to some subspace of the Hilbert space. \\
\end{tabular}
\end{center}
\caption[Overview of perturbative gadgets in the literature]{
    Overview of some relevant works related to perturbative gadgets. Not included are those that apply only to a restricted subset of Hamiltonians, such as topological Hamiltonians.
    }
\label{table:gadgets-overview}
\end{table*}

Perturbative gadgets~\cite{kempeComplexityLocalHamiltonian2006,oliveiraComplexityQuantumSpin2008a,jordanPerturbativeGadgetsArbitrary2008a,BrellGadget} encode the low-energy subspace of a many-body Hamiltonian into the low-energy subspace of a fewer-body Hamiltonian. 
They make use of perturbation theory in the opposite direction than what is more common: 
instead of considering the effect of some perturbation $\lambda V$ on a well-understood Hamiltonian $H$ according to $H' = H + \lambda V$, perturbative gadgets provide a suitable $H$ and perturbation $V$ such that $H'$ is a few-body Hamiltonian that approximates the low-energy subspace of a given, target Hamiltonian $H^{\text{target}}$, as visualized in Figure~\ref{fig:spectrum_embedding}. 
They are often accompanied by an increase in the number of required qubits and either suppress the energy of the desired low-energy subspace or an increase in the norm of $H'$ compared with $H^{\text{target}}$.

In most cases, a local and simple Hamiltonian with degenerate ground space is used as the so-called ``unperturbed Hamiltonian'' $H$. 
The degeneracy is chosen to correspond to the dimension of the space acted upon by the target Hamiltonian. 
Then, a perturbation is designed to split said degeneracy in a way that the resulting low-energy subspace emulates the target spectrum, so roughly
\begin{equation}
    \Pi H' \Pi \sim H^{\text{target}},
\end{equation}
where $\Pi$ denotes a projector on the low-energy subspace.
This statement will later be made formal in Theorem~\ref{theorem:main-result}.

To be explicit and clear, we use \textit{interaction weight} to refer to the number of particles acted upon by a given operator and \textit{interaction strength} to refer to the coefficient of that operator. 
When talking about qubits and expanding operators in the Pauli basis, the \textit{weight}
of a Pauli string (a tensor product of Pauli operators) is defined as the number of non-trivial Pauli operations in the string.
Throughout this work, we use ``locality'' to refer to the maximum weight of all terms in a Hamiltonian, independently of geometry.
We use the term \emph{local} for Hamiltonians with fixed small weights, that is weights that do not scale with problem size, while \emph{global} is reserved for Hamiltonians with all-to-all interactions.

An overview of works related to perturbative gadgets is presented in Table~\ref{table:gadgets-overview}.
Historically, perturbative gadgets were developed to prove the QMA-hardness of the local Hamiltonian problem. 
The $3$-local Hamiltonian problem is known to be QMA-complete~\cite{kempe3-LocalHamiltonianIsQMA2003} and Kempe, Kitaev, and Regev have shown that so is the $2$-local problem by proposing the first perturbative gadget, which reduces the locality from $3$ to $2$~\cite{kempeComplexityLocalHamiltonian2006}. 
Oliviera and Terhal have added to this by showing that the $2$-local Hamiltonian problem is still QMA-complete when restricting the interaction to a 2D square lattice, and in doing so introduced a new perturbative gadget. 
This subdivision gadget reduces the locality of a $k$-local Hamiltonian to $\lceil k/2\rceil+1$-local interactions by introducing a single auxiliary qubit per interaction term.
These kinds of gadgets have also been considered to be used outside of the world of complexity-theoretic analysis, but they are inherently not practical, as discussed further in Section~\ref{sec:practical-applications}.
\begin{figure}
    \centering
    \includegraphics{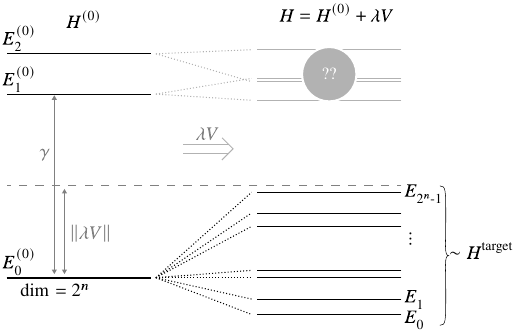}
    \caption[Perturbative gadgets: from unperturbed to target spectrum]
    {Sketch of the main idea behind perturbative gadgets. Starting from an unperturbed Hamiltonian $H^{(0)}$ with gap $\gamma$ and degenerate ground space, add a perturbation $V$ such that the lowest energy states resulting from the perturbation of the ground space mimic the target Hamiltonian $H^\text{target}$. 
    }
\label{fig:spectrum_embedding}
\end{figure}

Furthermore, the recursive use of gadgets to arrive at two-local Hamiltonians can be avoided by a single gadget introduced by Jordan and Farhi in Ref.~\cite{jordanPerturbativeGadgetsArbitrary2008a}, which encodes $k$-local Hamiltonians in $k^{\text{th}}$ order of perturbation theory of a two-local gadget Hamiltonian. 
However, just like the three-to-two local gadget, their construction relies on certain properties of the evolution under a Hamiltonian where the state is restricted to a predetermined subspace
, such as provided in the adiabatic setting.

In the following years, several works proposed improvements in resource requirements, convergence bounds, or coupling strengths \cite{Cao2015,caoCombinatorialAlgorithmsPerturbation2016,caoEfficientOptimizationPerturbative2017,bauschGadgets2020}, as well as introduced highly specialized gadget constructions for restricted classes of Hamiltonians, such as for topological quantum codes
and for realizing certain parent Hamiltonians featuring intrinsic topological order
\cite{BrellGadget,ockoNonperturbativeGadgetTopological2011}.

There are two important things to note at this point: While all of these gadget constructions rely on perturbation theory, they use different techniques, ranging from perturbative expansion due to Bloch~\cite{blochTheoriePerturbationsEtats1958}, which underlies the Jordan-Farhi gadget, the Feynman-Dyson series employed in 
Ref.~\cite{cubittUniversalHamiltonians2018}, to the Schrieffer-Wolf transformation~\cite{bravyiSchriefferWolffTransformation2011}, 
discussed in Ref.~\cite{caoEfficientOptimizationPerturbative2017}.
Furthermore, except for the so-called subdivision gadget due to Oliveira and Terhal~\cite{oliveiraComplexityQuantumSpin2008a}, the other gadgets are restricted to the setting of time evolution under the Hamiltonian, which guarantees the evolution to remain in the initialized subspaces, rendering them inapplicable for gate-based quantum computing. A non-recursive $k$-to-two-local gadget without these restrictions is currently unknown.

\section[Our gadget]{A direct \texorpdfstring{$k$}{k}-to-three-body gadget without subspace restriction}
\label{sec:new-gadget}

One of the main results of this work is the proposal of a direct, i.e., non-recursive, three-body gadget Hamiltonian $H^{\text{gad}}$ derived from an arbitrary $k$-body Hamiltonian $H^{\text{target}}$ that does not require any restrictions to a subspace of the Hilbert space, inspired by the work of Jordan and Farhi~\cite{jordanPerturbativeGadgetsArbitrary2008a}.
We then find that the low-energy subspace of $H^{\text{gad}}$ mimics the low-energy subspace of $H^{\text{target}}$ (Theorem~\ref{theorem:main-result}), implying a guarantee on the closeness of their ground states (Corollary~\ref{corollary:minimas}).

\begin{figure}
    \centering
    \includegraphics{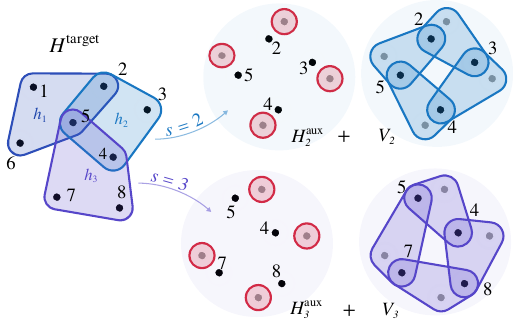}
    \caption{Representation of a target Hamiltonian (left) and its associated gadget Hamiltonian, as defined in Eq.~\eqref{eq:def-Hgad}, for the case of $n=8$, $k=4$, and $r=3$. Qubits that interact according to the terms of the Hamiltonian are placed in the same colored region. Shown in the central and right-hand panels are the contributions to the gadget Hamiltonian from two of the three four-body terms, with the auxiliary qubits in gray and the target qubits in black.}
    \label{fig:terms}
\end{figure}

We start by defining the specific gadget Hamiltonian that is the focus of our study.

\begin{definition}[Gadget Hamiltonian]\label{def:hamiltonians}
Let $H^{\text{target}}$ be a $k$-body Hamiltonian acting on $n$ qubits, given by
\begin{equation}\label{eq:def-Hcomp}
    H^\text{target} \coloneqq \sum\limits_{s=1}^r c_s h_s, \quad \text{with} \quad h_s = \sigma_{s,1} \otimes \sigma_{s,2} \otimes \dots \otimes \sigma_{s,k},
\end{equation}
where each term is a tensor product of at most $k\leq n$ single qubit operators, with $\sigma_{s,j}=n_X X+n_YY+n_ZZ$ and $(n_X,n_Y,n_Z)\in\mathbb{R}^3$ a unit vector.
We define the three-body gadget Hamiltonian $H^{\text{gad}}$ corresponding to $H^{\text{target}}$, acting on $n+rk$ qubits, as
\begin{equation} \label{eq:def-Hgad}
    H^{\text{gad}} \coloneqq \sum\limits_{ s = 1}^r H^\text{aux}_s + \lambda \sum\limits_{ s = 1}^r V_s,
\end{equation}
where
\begin{align} 
    H^\text{aux}_s &\coloneqq \sum\limits_{ j = 1}^k \frac{1}{2} \left( \II_{s,j} - Z_{s,j}^{\text{aux}} \right)
    = \sum\limits_{ j = 1}^k \ketbra{1}{1}_{s,j}^{\text{aux}}, \label{eq:def-Hanc} \\
    V_s &\coloneqq \sum\limits_{j=1}^k \tilde{c}_{s,j} \sigma_{s,j}^\text{target} \otimes X_{s,j}^\text{aux} \otimes X_{s,(j+1)\text{ mod }k}^\text{aux} \label{eq:def-V}
\end{align}
and $\tilde{c}_{s,j}=-(-1)^k c_s$ if $j=1$ or $\tilde{c}_{s,j}=1$ otherwise.
We refer to the $n$ qubits on which $H^\text{target}$ acts as \emph{target qubits}, and the $rk$ additional qubits as \emph{auxiliary qubits}.
\end{definition}

In our gadget Hamiltonian, each term $H_s^\text{aux}$ acts only on the $s^{\text{th}}$ group of auxiliary qubits, while each term $V_s$ acts on a target qubit and the $s^{\text{th}}$ group of auxiliary qubits; see Fig.~\ref{fig:terms} for a graphical depiction for a toy model.

\begin{figure*}
    \centering
    \includegraphics[width=\textwidth]{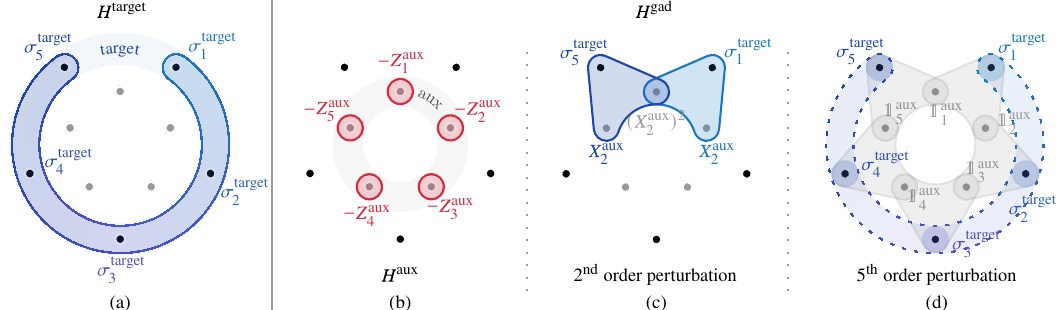}
    \caption{Visualization of the key steps in the proof of Theorem~\ref{theorem:main-result}. (a) For illustrative purposes, we take as the $k$-body target Hamiltonian $H^\text{target} = \bigotimes_{j=1}^5 \sigma_j$, where $\sigma_j\in\{X,Y,Z\}$ are Pauli operators, so that $k=5$. (b) The unperturbed auxiliary Hamiltonian, given in Eq.~\eqref{eq:def-Hanc}, acts only on the auxiliary qubits. (c) Acting with the perturbation $V$ in Eq.~\eqref{eq:def-V}, at second order, one term could result in flipping the second and last auxiliary qubits, hence ejecting the state out of the ground space and not contributing to the perturbative expansion. (d) At $k^{\text{th}}$ order in perturbation theory, we obtain a non-trivial contribution acting on all target qubits simultaneously, mimicking the target Hamiltonian, while acting trivially on the auxiliary qubits.
    }
    \label{fig:gadget-analysis}
\end{figure*}

The proposed gadget is heavily inspired by the gadget introduced by Jordan and Farhi in Ref.~\cite{jordanPerturbativeGadgetsArbitrary2008a}, in which they rely on properties of adiabatic quantum computing.
Both their gadget and ours take advantage of a larger Hilbert space to encode the low-energy subspace of the target Hamiltonian, so the dimension of the ground space of the unperturbed Hamiltonian must be the same as the space on which the target Hamiltonian acts.
This point is made more explicit in Appendix~\ref{A:proof}.
Jordan and Farhi start from a larger ground space of the unperturbed Hamiltonian and reduce the dimension by restricting the region of the Hilbert space that can be explored by initializing the system in a fixed subspace. 
Remaining in this chosen subspace is then ensured by using properties of adiabatic evolution, something that cannot be done straightforwardly outside of the regime of adiabatic quantum computing. 
This requirement, therefore, places a strong limitation on the applicability of their construction. 
On the other hand, our gadget construction does not require such a restriction to a subspace, therefore avoids this limitation, and can be applied even in the context of digital, gate-based quantum computing.
We refer to Appendix~\ref{A:Farhi-inapplicability} for a more detailed discussion of why a direct application of the gadget introduced by Jordan and Farhi is not possible in the realm of non-adiabatic quantum computing.

Similar to the results of Jordan and Farhi, we can show that the subspace of the $2^n$ lowest energy eigenstates of our three-body gadget Hamiltonian in Eq.~\eqref{eq:def-Hgad} mimics that of the $k$-body target Hamiltonian. 
For the formal statement, we define the following notation: Let $H$ be a Hamiltonian acting on $n$ qubits, with spectral decomposition $H=\sum_{j=0}^{2^n-1}E_j\ketbra{\psi_j}{\psi_j}$ such that $E_0\leq E_1\leq\dotsb\leq E_{2^n-1}$, we define $H_{\text{eff}}(H,d)=\sum_{j=0}^{d-1}E_j\ketbra{\psi_j}{\psi_j}$ {to} be the effective Hamiltonian corresponding to the $d$ lowest eigenvalues of $H$.

\begin{theorem}[Main result] \label{theorem:main-result}
Let $H^\text{target}$ and $H^\text{gad}$ be as in Definition~\ref{def:hamiltonians}, and let $\lambda \leq \lambda_{\text{max}}$, with $\lambda_{\text{max}} = \frac{1}{4}\left(\sum_{s=1}^r|c_s|+r(k-1)\right)^{-1}$. Then, there exists an $f(\lambda)=\mathcal{O}(\operatorname{poly} \lambda)$ and $\Xi=\mathcal{O}(\operatorname{poly} k)$ such that
\begin{multline} \label{eq:Heff}
    H_\text{eff}(H^{\text{gad}}, 2^{n})
    = \frac{\lambda^k}{\Xi} H^\text{target} \otimes \left(|0 \rangle\!\< 0 |\right)^{\otimes rk}
    \\+ f(\lambda)\Pi + \mathcal{O}(\lambda^{k+1}),
\end{multline}
where $\Pi$ is the projector onto the support of $H_\text{eff}(H^{\text{gad}},2^n)$.
\end{theorem}

Furthermore, we can provide a guarantee on the closeness of the ground states of the target and gadget Hamiltonians.

\begin{corollary}[Guarantees on the closeness of the ground states]\label{corollary:minimas}
Let $H^\text{target}$ and $H^\text{gad}$ be as in Theorem~\ref{theorem:main-result}. Then, there exists a perturbation strength $\lambda^*$, with 
\begin{equation}
    \lambda^*\leq\lambda_{\text{max}}
    \quad \text{ and } \quad
    \lambda^* \leq \frac{E^\text{target}_1 - E^\text{target}_0 }{\Xi \|O^\text{err} \|},
\end{equation}
such that for all $\lambda \leq \lambda^*$ it holds that
\begin{equation} \label{eq:gadget_GS_bound}
    \left\lVert \psi_0 - \Tr_{\text{aux}}[\phi_{0}] \right\rVert_2 \leq \mathcal{O}(\lambda),
\end{equation}
where $\psi_0=\ketbra{\psi_0}{\psi_0}$ and $\phi_0=\ketbra{\phi_0}{\phi_0}$ are in the ground spaces of $H^\text{target}$ and $H^\text{gad}$, respectively.
\end{corollary}

In Fig.~\ref{fig:gadget-analysis}, we provide a visualization of the main steps in the proof of Theorem~\ref{theorem:main-result},  
focusing on the main ideas behind these results and referring the mathematically interested reader to the appendices. 
Specifically, after a recap in Appendix~\ref{A:Bloch-expansion} of the perturbation theory on which our results are based, we prove Theorem~\ref{theorem:main-result} in Appendix~\ref{A:proof} and Corollary~\ref{corollary:minimas} in Appendix~\ref{A:proof-corr}.

For clarity, the construction is stated for target Hamiltonians where each term has the same Pauli weight, and for the maximal reduction: to a three-local gadget Hamiltonian. 
This gadget can be generalized to both of these aspects and we show how to construct a gadget for mixed Pauli weights in Appendix~\ref{A:extension1} and also how to trade off target locality for reduced resources in Appendix~\ref{A:extension2}.

\section[General recipe]{A recipe for creating your own perturbative gadget}
\label{sec:recipe}

In the previous section, we presented a specific perturbative gadget construction. 
However, that construction is by no means unique, because there is some freedom in designing such gadgets. 
For instance, the choice of Pauli operators in the unperturbed Hamiltonian acting on the auxiliary qubits or in the perturbation could have been different. 
More generally, the general construction of the penalization in $H^\text{aux}_s$ and of the perturbation $V_s$ leaves many knobs to tweak. 
Especially for practical implementations, it might be interesting to derive slightly different gadgets, which result in similar low-energy spectra but are constructed from a different set of operators, for instance, some that are easier to implement on the hardware of choice.

Here, we present a more general framework for the construction of non-recursive, non-adiabatic perturbative gadgets based on higher-order perturbation theory that all lead to similar results as Theorem~\ref{theorem:main-result}, with the only notable difference being the exact form of $\Xi$ in Eq.~\eqref{eq:Heff}.

The core idea is again to start from a well-known, unperturbed Hamiltonian $H^\text{aux}$ and to perturb it with a perturbation $\lambda V$, such that we recover the original, $k$-body target Hamiltonian at  $k^{\text{th}}$ order in perturbation theory.
Since the $k^\text{th}$ order in perturbation theory comes with $k$-fold applications of the perturbation $V$, we engineer the perturbation such that only cross-terms leading to the desired result contribute. 
That is if the target Hamiltonian has been cut into $k$ pieces $O^\text{target} = O_1 \otimes \ldots \otimes O_k$, we ensure that only the correct $k$-fold products of $V$ survive. 

To construct a perturbative gadget following a similar recipe as the one presented in this work, we require a \emph{penalization Hamiltonian} $H$ and a \emph{perturbation operator} $A$, both acting on a single auxiliary register of $k$ qubits.
Although those are not the unperturbed Hamiltonian $H^\text{aux}$ and the perturbation $V$ directly, they are closely related. $H^\text{aux}$ will be built using instances of $H$, while the perturbation operator will be a key component in the construction of the perturbation $V$.
We further require $H^\text{aux}$ to be composed of few-body terms and to have a non-degenerate ground space to avoid the limitations of prior gadget constructions~\cite{kempeComplexityLocalHamiltonian2006,subasiNonperturbativeEmbeddingHighly2016,jordanPerturbativeGadgetsArbitrary2008a}.
That is, we need
\begin{equation}\label{eq-penalization_H}
    H=\sum_i h_i,
\end{equation}
for some finite sum of operators $h_i$, each of which is few-body, such that $H$ has a unique ground state vector, 
which we denote by $| GS \rangle $.
We note that the restriction on few-body terms is not strictly required for the construction to yield the desired low-energy subspace, but for decreasing the locality of the Hamiltonian.

Let us now argue that the operator $A$ should exhibit a product form of $k$ operators $a_1,a_2,\dotsc,a_k$, such that
\begin{equation}\label{eq-perturbation_A}
    A = \prod_{j=1}^k a_j,
\end{equation}
with the following properties:
\begin{align} \label{Aeq:conditions_A}
    & A |GS\rangle \propto |GS\rangle, \\
    & \< GS | \prod_{\ell=1}^m a_{j_{\ell}} |GS\rangle = 0 \quad \forall~ m<k,\, \forall~\{j_{\ell}\}\subset\{1,2,\dotsc,k\}, \\
    & a_j^2 = \II \quad \forall~j\in\{1,2,\dotsc,k\}.
\end{align}
The first condition states that the ground state of $H$ should also be an eigenstate of $A$, while the second enforces that any partial application of the operators $a_j$ results in a state orthogonal to $|GS\rangle$.
These are vital to ensure that at orders lower than $k$ in perturbation theory, the perturbation will always expel the state out of the ground space.
The final property ensures that when the same perturbation is applied twice, it results in a constant energy shift.

With the above properties, we can construct $H^{\text{aux}}$ and $V$ as
\begin{align}
    H^\text{aux}_s & = \underbrace{\II^{\otimes n}}_{\mathclap{\text{target register}}\hspace{1,0em}} \otimes \underbrace{ \II^{\otimes (s-1)k} \otimes H \otimes \II^{\otimes (r-s)k}}_{\text{auxiliary registers}}, \label{def:recipe-Haux}\\
    V_s & = \sum_{j=1}^k \tilde{c}_{s,j} \sigma_{s,j} \otimes \II^{\otimes (s-1)k} \otimes a_j \otimes \II^{\otimes (r-s)k},
\end{align}
where the definition of the coefficients $\tilde{c}_{s,j}$ provides another knob to tweak. They could be defined similarly to Definition~\ref{def:hamiltonians} or in an arbitrary fashion, as long as
\begin{equation}
    \prod_{j=1}^k \tilde{c}_{s,j} = -(-1)^k c_s
\end{equation}
still holds. 
Given these conditions, similar results as in Theorem~\ref{theorem:main-result} will hold due to the arguments laid out in Appendix~\ref{A:proof}. 
Indeed, the gadget we propose is simply a special case of
\begin{align}
    h_i & = \frac{1}{2} (\II - Z_i) = | 1 \rangle \! \< 1 |_i, & & 
    i\in\{1,2,\dotsc,k\},\\
    a_j & = X_j \otimes X_{(j+1) \text{ mod } k}, & & 
    j\in\{1,2,\dotsc,k\}.
\end{align}
Note that the Hamiltonian in Eq.~\eqref{def:recipe-Haux} has a non-degenerate ground space when considering only the auxiliary registers but has a $2^n$ degenerate ground space when additionally considering the target qubits which it acts trivially upon.

Furthermore, we would like to stress that any gadget constructed in this fashion will be accompanied by a suppression of the target Hamiltonian in the low-energy subspace by a factor of $\lambda^k$, therefore not avoiding a blow-up in resource costs exponential in the reduction in the locality.
Since the perturbative expansions only hold as long as the perturbation parameter $\lambda\leq \frac{1}{4}\left(\sum_{s=1}^r|c_s|+r(k-1)\right)$, this trade-off is ingrained into the construction and cannot be circumvented by rescaling the coefficients $c_s$ or similar.

\section[Practical applications]{On the trade-offs in using perturbative gadgets for practical applications}
\label{sec:practical-applications}

Perturbative gadgets have been a useful tool in complexity theory, in particular in the study of the QMA-completeness of the local Hamiltonian problem~\cite{kempeComplexityLocalHamiltonian2006, oliveiraComplexityQuantumSpin2008a}. 
These successes and the locality-reducing property encourage us to look for practical applications of these kinds of gadgets outside of the world of complexity-theoretic analyses.
The first candidate is adiabatic computing and analog quantum simulation~\cite{biamonteRealizableHamiltoniansUniversal2008, bravyiQuantumSimulationManyBody2008, cubittUniversalHamiltonians2018, harleyGoingGadgetsImportance2023}. 
Due to the limitations of experimental setups to implement few-body couplings, these simulations could benefit from the locality reduction to effectively study many-body Hamiltonians while physically implementing few-body interactions.
Unfortunately, the large differences in coupling strengths pose a challenge for experimental realizations~\cite{bravyiQuantumSimulationManyBody2008, bauschGadgets2020, harleyGoingGadgetsImportance2023}.
Having lifted the restriction to subspaces in the gadget construction limiting previous proposals, we can now turn towards gate-based quantum computing in the search for potential applications, an area that was previously not readily accessible. 
In the following, we look at four specific settings that could benefit from perturbative gadgets and discuss the accompanying trade-offs, which are usually unfavorable. 
While we set the stage for the potential application of our gadgets to these settings, we leave the detailed study of these applications to future work.

We start with variational quantum algorithms and discuss the issue of cost-function-dependent barren plateaus and the readout problem, before noting the impact of the locality of the measurement operator on measurement noise.
Afterward, we discuss possible applications in phase estimation and time evolution for restricted hardware settings.

\paragraph*{Cost-function-dependent barren plateaus.}
\emph{Variational quantum algorithms}  (VQAs)~\cite{cerezoVariationalQuantumAlgorithms2021a,bhartiNoisyIntermediatescaleQuantum2022,mccleanTheoryVariationalHybrid2016a} are strong in the run to find useful applications of quantum computers in the \emph{noisy intermediate scale quantum} (NISQ) era~\cite{preskillQuantumComputingNISQ2018a} and are being intensively studied.
These algorithms rely on \emph{parametrized quantum circuits} (PQCs) to evaluate a parameter-dependent cost function related to the expectation value of a set of observables. 
An optimization algorithm, e.g., gradient descent, implemented on a classical device is used to optimize the parameters of the PQC to find the minimum of the cost function. 

Two of the most prominent examples of VQAs might be the variational quantum eigensolver~\cite{peruzzoVariationalEigenvalueSolver2014a}, aimed at finding energies of Hamiltonians, and the quantum approximate optimization algorithm~\cite{farhiQuantumApproximateOptimization2014}.
One hurdle to overcome for such algorithms to become useful is the so-called barren plateau issue~\cite{mccleanBarrenPlateausQuantum2018a, wangNoiseInducedBarrenPlateaus2021,marreroEntanglementInducedBarren2021, holmesConnectingAnsatzExpressibility2022, martinBP22}. 
This refers to the phenomenon where the cost-function landscape becomes essentially flat, rendering optimization algorithms ineffective due to the lack of efficiently measurable gradients. 
Many strategies for mitigating barren plateaus have been proposed~\cite{grantInitializationStrategyAddressing2019, holmesConnectingAnsatzExpressibility2022,volkoffLargeGradientsCorrelation2021, wiersemaMeasurementinducedEntanglementPhase2021, meleAvoidingBarrenPlateaus2022, skolikLayerwiseLearningQuantum2021, kieferovaQuantumGenerativeTraining2021,dborinMatrixProductState2021, liuMitigatingBarrenPlateaus2022, sackAvoidingBarrenPlateaus2022a, kianiLearningQuantumData2022}, but the observation most appealing to us is the fact that, given two cost functions with the same minimum but different localities, the one with reduced locality performs better~\cite{Khatri2019,laroseVQSD,prietoVQLS}. 
This dependence of the emergence of barren plateaus on the locality of the cost function was formalized for some ansätze in Refs.~\cite{cerezoCostFunctionDependent2021a,uvarovBarrenPlateausCost2021a}.

Looking back at the results from Section~\ref{sec:new-gadget}, our gadget enables a generic procedure for any cost function written as the expectation value of a Hamiltonian to find an equivalent $3$-local cost function.
This new cost function has the same minimum (Corollary~\ref{corollary:minimas}) as the target cost function and has provably non-vanishing gradients (due to Ref.~\cite[Theorem~2]{uvarovBarrenPlateausCost2021a}). 

Upon first inspection, the problem of cost-function-locality-dependent barren plateaus seems to be solved, but one needs to take a step back and remember the original problem. 
Barren plateaus are an issue because exponentially vanishing gradients require an exponential number of measurements on a quantum device to resolve the gradients. 
This experimental cost means that such procedures are not scalable. 
Although on the surface our proposal produces gradients that can be computed efficiently, finding the minimum of the \emph{effective} gadget Hamiltonian might become hard. Indeed, let us consider the extreme case of a global cost function, i.e., an $n$-local cost function.
In this setting, from Theorem~\ref{theorem:main-result}, the contribution of the target Hamiltonian is suppressed as $\lambda^n$. 
In the end, the features of interest from the target Hamiltonian might become exponentially small.
Consequently, although the true minimum of the gadget cost function coincides with the minimum of the target cost function, the actual optimization may converge very close to the minimum of the gadget, but owing to the exponential suppression the energy of the target may still be very high.

Not all is lost, and we show some successful simulations in Appendix~\ref{A:numerics}. 
Our gadget-induced cost function could be used as an initialization strategy. 
Indeed, even if the target Hamiltonian is exponentially suppressed in the effective low-energy subspace, each of the terms of the perturbation in gadget Hamiltonian contains some part of the target Hamiltonian. 
It could be that using our gadget provides a useful path to reach a region of the target cost function that is outside of the regime plagued by barren plateaus. 
Furthermore, previous works on perturbative gadgets argue that restricting to the regime where the perturbative expansion converges is not necessarily justified~\cite{bravyiQuantumSimulationManyBody2008}. 
In our case, that means that it could be beneficial to choose a perturbation parameter $\lambda$ larger than the bound for which our analytical results hold, thus increasing the contribution of the target Hamiltonian. 
This idea is supported by the results of our numerical simulations.
Interpreting this perturbation factor as a hyperparameter of the model, one could take inspiration from classical machine learning and implement a schedule of slowly decreasing $\lambda$, having strong contributions at the beginning and high precision at the end.
These heuristics will however need to be confirmed by larger experiments than those we could perform or by further theoretical studies. 
An extensive description of the problem of cost-function-dependent barren plateaus and of the attempt to use perturbative gadgets, complemented with numerical simulations, can be found in Appendix~\ref{A:barren-plateaus}.

\paragraph*{Reducing measurement bases for readout.}

Another, albeit more pessimistic example of possible trade-offs using perturbative gadgets is the construction of a gadget aimed at exploiting the commutation structure of the gadget Hamiltonian.

As for the gadget introduced in Eq.~\eqref{eq:def-Hgad}, by first measuring all auxiliary qubits in the Pauli-$Z$ basis, we can estimate $H^\textrm{aux}$.
Then, consecutively measuring all auxiliary qubits in the Pauli-$X$ basis and simultaneously all target qubits in one of the three Pauli bases, we can estimate all of the $V_s$ terms. Consequently, only four different measurements are required to obtain an unbiased estimator of the energy of the gadget Hamiltonian.

Going one step further, we could design a new gadget that decomposes each term of the Hamiltonian into its Pauli-$X$, Pauli-$Y$, and Pauli-$Z$ parts and use each part in a single perturbation term. That would lead to a gadget that reconstructs the original term in the third order of perturbation theory and thereby reduces the complexity of the gadget.

In case the locality is of no importance, we can also reduce the number of terms in the gadget Hamiltonian and the qubit overhead. Through a different definition of the part of the perturbation in the gadget Hamiltonian that acts on the auxiliary register, we can achieve only a logarithmic qubit overhead, as shown in more detail in Appendix~\ref{A:measurement_gadget}.

However, we still end up facing the same challenge as before: While we can efficiently estimate the energy of the gadget Hamiltonian, that is not the quantity we are interested in.
Instead, we want to measure the energy of the target Hamiltonian, which is suppressed by $\lambda^3$ and will, due to the implicit dependence of $\lambda$ on the number of terms in the target Hamiltonian, lead to an unfavorable scheme compared to simply measuring each term of the Hamiltonian by itself.
The only potential use case for such a gadget would then be a setting where it is excessively more expensive to change the measurement setting than to evaluate the circuit.

\paragraph*{Measurement noise.}
Another setting in which the locality of certain operators can noticeably impact the performance of a gate-based quantum computer is measurement noise.
As shown in the context of noise-induced barren plateaus~\cite{wangNoiseInducedBarrenPlateaus2021}, measurement noise scales directly with the locality of the measurement operator.
Therefore, even when the locality of the cost function is constant in the number of qubits, even constant reductions in locality might improve the performance of the (variational) algorithm.
It is important to note that this reduction will still be accompanied by energy suppression, so whether such a trade-off is favorable again highly dependent on the actual setting.
However, it is an application within gate-based quantum computing, that could previously not be addressed by perturbative gadgets due to the subspace restrictions of previous constructions or their recursive nature overcomplicating implementation.

\paragraph*{Phase estimation.}
The final example of a situation that could benefit from reductions in locality is that of quantum phase estimation.
There, controlled time evolutions for some Hamiltonian $H$ are combined with the quantum Fourier transform to estimate the energies of quantum states.
Since approximating the time evolution, e.g., by product formulas requires the application of quantum gates whose locality scales with the locality of the Hamiltonian, the feasibility of such algorithms depends on the used hardware platform.
Appropriate perturbative gadgets could then be used to circumvent limitations in hardware connectivity or reduce the number of required entangling gates.
While the accompanying trade-off in the energy scale would result in improved required accuracy of its readout and consequently larger circuit depths and auxiliary qubit counts, they might be favorable in certain settings.
Also, since quantum phase estimation requires the preparation of an initial state with a large overlap of the true ground state (if the goal is to estimate the ground state energy), using perturbative gadgets does not introduce additional complications, because the previous initial state can just be extended by initializing all auxiliary qubits of the gadget in the zero state.

\section{Summary and outlook}
\label{sec:outlook}

Perturbative gadgets are a powerful tool that can be used to show equivalences between Hamiltonians of different localities. 
They have led to advances in quantum complexity theory, 
addressing problems that originate from studies of quantum many-body systems in rigorous condensed-matter theory, 
and have been considered tools to facilitate feasible implementations in experimental settings. 
They have also been explored as tools to devise schemes of quantum simulations of systems that are natively too difficult to tackle,
say, quantum systems featuring certain kinds of intrinsic topological order \cite{BrellGadget,Wille}.
At the same time, perturbative gadgets are quite often accompanied by unfavorable trade-offs in interaction strengths, recursions, or limitations to settings allowing for subspace restrictions, bringing into question their use for practical applications~\cite{harleyGoingGadgetsImportance2023}.

In this work, we have proposed a type of perturbative gadget that jointly lifts the requirements on subspace restrictions and recursions but cannot circumvent the problem of impractical interaction strengths.
It achieves the same locality reduction as the subdivision gadget from Oliveira and Terhal~\cite{oliveiraComplexityQuantumSpin2008a}, but has a direct formulation in higher-orders of perturbation theory as does the gadget from Jordan and Farhi~\cite{jordanPerturbativeGadgetsArbitrary2008a}, without requiring to restrict the evolution of the system to a subspace of the Hilbert space. 
In doing so, we observe that each of these constructions is a special case and demonstrate how to generalize our construction to a larger class of gadgets applicable to gate-based quantum computing.

Furthermore, we turn towards the setting of variational quantum algorithms, where the locality of cost functions can be of great interest and thus benefit from generic procedures producing Hamiltonians with the same ground states but reduced locality.
Also for fault-tolerant quantum algorithms, reductions in locality could help alleviate hardware constraints or impact measurement noise.

We note that locality versus energy scale trade-offs in using perturbative gadgets may still be impractical outside of theoretical considerations and therefore leave the question of their practicality wide open.
Circumventing the need for subspace restrictions and recursions does, however, broaden their scope and therefore opens up new directions of exploration for finding practical applications.
It might be interesting to consider the question of how imposing restrictions on the target Hamiltonian might allow us to improve the resource requirements of our gadget construction. 
In the same flavor as gadgets for topological quantum codes, there might be classes of Hamiltonians for which the effective Hamiltonian can be retrieved at a lower perturbation order. 
Such a development would trade off universality for performance, an idea established in the area of variational quantum algorithms and beyond.

\paragraph*{Author contributions.} 
PKF has conceived the project. The theoretical analysis has been developed by SC, PKF, and SK, and the numerical demonstrations have been developed by SC under the supervision of PKF and SK. JE has supported research and development. All authors have contributed to
scientific discussions and to the writing of the manuscript.\smallskip

\paragraph*{Acknowledgements.} 
We thank Johannes Jakob Meyer, Jakob Kottmann, and Lennart Bittel for helpful discussions and for providing feedback on an early version of the manuscript. Financial support has been provided by the BMWK (PlanQK, EniQmA),
the QuantERA (HQCC),
the BMBF (Hybrid), the 
Munich Quantum Valley 
(K-8), the ERC (DebuQC), and the 
Einstein Foundation (Einstein Research 
Unit on Quantum Devices).\smallskip

\paragraph*{Code availability.} 
The numerical simulations were performed using the cross-platform Python library for differentiable programming of quantum computers PennyLane~\cite{bergholmPennyLaneAutomaticDifferentiation2022} and the code is available at \href{https://github.com/SimonCichy/barren-gadgets}{https://github.com/SimonCichy/barren-gadgets}.

\bibliography{mybibliography.bib}
\stopcontents[mainsections]
\newpage
\onecolumngrid

\newpage
\onecolumngrid
\begin{center}
\textbf{\large Appendix for ``Perturbative gadgets for gate-based quantum computing: Non-recursive constructions without subspace restrictions''}\\
\vspace{1ex}
\end{center}
\appendix

\startcontents[appendices]
\printcontents[appendices]{l}{1}{}

\section{Review of perturbation theory}
\label{A:Bloch-expansion}
In this section, we briefly summarize relevant results from perturbation theory. In particular, we closely follow the presentation in Ref.~\cite{blochTheoriePerturbationsEtats1958} for obtaining a general perturbative expansion, which we need for the Proof of Theorem~\ref{theorem:main-result} in Appendix~\ref{A:proof}.

\subsection{Definitions}
\label{A:Definitions}
The goal of perturbation theory is to calculate the effect of a small perturbation $V$ of strength $\lambda\geq 0$ on a known, unperturbed Hamiltonian $H^{(0)}$.
In particular, the goal is to study the Hamiltonian
\begin{equation}
    H = H^{(0)} + \lambda V
\end{equation}
and to determine its low-energy eigenvalues and corresponding eigenvectors as a function of $\lambda$. To this end, suppose that $H^{(0)}$ and $V$ act on a $D$-dimensional Hilbert space. Let $E_0^{(0)}$ be the lowest eigenenergy of $H^{(0)}$, and let
\begin{equation}
    \mathcal{E}^{(0)}\coloneqq \textrm{span}\left\{ |\varphi\rangle \in\mathbb{C}^D : H^{(0)}|\varphi\rangle = E^{(0)}_0|\varphi\rangle \right\}
\end{equation}
be the corresponding $d$-dimensional (degenerate) ground space of $H^{(0)}$, $d\leq D$. We let $\{\ket{j}\in\mathbb{C}^D\}_{j=1}^d$ denote an orthonormal basis for $\mathcal{E}^{(0)}$. We also let $\ket{\psi_1},\ket{\psi_2},\dotsc,\ket{\psi_d}\in\mathbb{C}^D$ be the perturbed orthonormal eigenvector (i.e., the eigenvectors of $H$) corresponding to the $d$-dimensional subspace $\mathcal{E}^{(0)}$, such that
\begin{equation}\label{A-eq:intermediate}
    H|\psi_j\rangle = \widetilde{E}_j |\psi_j\rangle \quad \forall~j\in\{1,2,\dotsc,d\},
\end{equation}
where $\widetilde{E}_j$ are the energies corresponding to $\ket{\psi_j}$. For the rest of this section, we are concerned with the shifted energies $E_j\coloneqq \widetilde{E}_j-E_0^{(0)}$. We also let
\begin{equation}
    \mathcal{E}\coloneqq\text{span}\left\{\ket{\psi_j}\in\mathbb{C}^D:j\in\{1,2,\dotsc,d\}\right\}
\end{equation}
be the $d$-dimensional subspace spanned by the vectors $\ket{\psi_1},\ket{\psi_2},\dotsc,\ket{\psi_d}$. We assume throughout that the subspaces $\mathcal{E}^{(0)}$ and $\mathcal{E}$ are not orthogonal, meaning that there does not exist a vector in $\mathcal{E}^{(0)}$ that is orthogonal to all vectors in $\mathcal{E}$, and vice versa. This condition can be guaranteed if the perturbation strength is not too high~\cite{blochTheoriePerturbationsEtats1958}.

Our specific goal in this section is to determine a perturbative expansion of $H$ in the $d$-dimensional subspace $\mathcal{E}$, i.e., we seek a perturbative expansion of the \emph{effective Hamiltonian}
\begin{equation}
    H_{\text{eff}}(H,d)\coloneqq\sum_{j=1}^d E_j\ketbra{\psi_j}{\psi_j}.
\end{equation}
Note that this is nothing more than the projection of the Hamiltonian $H$ onto the subspace $\mathcal{E}$, followed by a constant shift of $E_0^{(0)}$.

Let $P_0$ be the projector onto $ \mathcal{E}^{(0)}$. We can thus write $H^{(0)}$ as
\begin{equation}\label{eq:H_unpert}
    H^{(0)}=E_0^{(0)}P_0+Q,
\end{equation}
where $Q$ satisfies $P_0Q=QP_0=0$. Let us also define $ |\alpha_j\rangle $ to be the (non-normalized) state vectors reflecting the projections of the perturbed eigenstates onto the subspace $\mathcal{E}^{(0)}$, i.e.,
\begin{equation}\label{eq-proj_vectors}
    |\alpha_j\rangle \coloneqq P_0|\psi_j\rangle . 
\end{equation}
Note that these vectors are not necessarily orthonormal. But they are linearly independent, as we now prove.

\begin{lemma}[Linear independence]\label{lem-proj_vectors}
    Under the assumption that the vector spaces $\mathcal{E}^{(0)}$ and $\mathcal{E}$ are not orthogonal, the vectors defined in Eq.~\eqref{eq-proj_vectors} are linearly independent. Consequently, there exist vectors $\ket{\widetilde{\alpha}_j}$, $j\in\{1,2,\dotsc,d\}$, such that $\braket{\alpha_j}{\widetilde{\alpha}_{j'}}=\delta_{j,j'}$ for all $j,j'\in\{1,2,\dotsc,d\}$ and
    \begin{equation}\label{eq-proj_E0}
        P_0=\sum_{j=1}^d\ketbra{\alpha_j}{\widetilde{\alpha}_j}=\sum_{j=1}^d\ketbra{\widetilde{\alpha}_j}{\alpha_j}.
    \end{equation}
\end{lemma}

\begin{proof}
    Consider the equation $c_1\ket{\alpha_1}+c_2\ket{\alpha_2}+\dotsb+c_d\ket{\alpha_d}=0$, where $c_1,c_2,\dotsc,c_d\in\mathbb{C}$. Using the definition of $\ket{\alpha_j}$, we obtain
    \begin{equation}
    P_0\left(c_1\ket{\psi_1}+c_2\ket{\psi_2}+\dotsb+c_d\ket{\psi_d}\right)=0.
    \end{equation}
    Now, because $\mathcal{E}^{(0)}$ and $\mathcal{E}$ are non-orthogonal, it follows that the vector $c_1\ket{\psi_1}+c_2\ket{\psi_2}+\dotsb+c_d\ket{\psi_d}$ is not in the orthogonal complement of $\mathcal{E}^{(0)}$, which means that it must be equal to the zero vector. Consequently, because the vectors $\ket{\psi_j}$ are linearly independent, we have that $c_1=c_2=\dotsb=c_d=0$. Hence, because the numbers $c_1,c_2,\dotsc,c_d$ were arbitrary, we conclude that the vectors $\ket{\alpha_j}$ are linearly independent.
    
    Now, consider the operator $T\coloneqq\sum_{j=1}^d\ketbra{\alpha_j}{j}$, which can be thought of as a matrix whose columns are equal to the vectors $\ket{\alpha_j}$. In particular, note that $\ket{\alpha_j}=T\ket{j}$ for all $j\in\{1,2,\dotsc,d\}$. Linear independence of the vectors $\ket{\alpha_j}$ implies that $T$ is invertible. (Here, and throughout this section, by \emph{inverse} we mean the Moore--Penrose pseudo-inverse, which is the inverse on the support of a linear operator.) Hence, letting
    \begin{equation}\label{eq-proj_vectors_tilde}
        \ket{\widetilde{\alpha}_j}\coloneqq (T^{-1})^{\dagger}\ket{j},\quad\forall~j\in\{1,2,\dotsc,d\},
    \end{equation}
    we find that
    \begin{equation}
        \braket{\alpha_j}{\widetilde{\alpha}_{j'}}=\bra{j}T^{\dagger}(T^{-1})^{\dagger}\ket{j'}=\braket{j}{j'}=\delta_{j,j'}\quad\forall~j,j'\in\{1,2,\dotsc,d\},
    \end{equation}
    as required.
    
    Finally, using Eq.~\eqref{eq-proj_vectors_tilde} and the fact that $T\ket{j}=\ket{\alpha_j}$ for all $j\in\{1,2,\dotsc,d\}$, we find that both $\sum_{j=1}^d\ketbra{\alpha_j}{\widetilde{\alpha}_j}$ and $\sum_{j=1}^d\ketbra{\widetilde{\alpha}_j}{\alpha_j}$ are equal to the identity operator for the subspace $\mathcal{E}^{(0)}$, which is precisely equal to the projection $P_0$. This completes the proof.
\end{proof}

\begin{lemma}[Invertibility]
    Consider the two linear operators
    \begin{align}
        \mathcal{U}&\coloneqq \sum_{j=1}^d \ketbra{\psi_j}{\widetilde{\alpha}_j},\label{eq-pert_U}\\
        \mathcal{A}&\coloneqq \lambda P_0 V\mathcal{U}.
    \end{align}
    The operator $\mathcal{U}$ is invertible on the subspace $\mathcal{E}$, and
    \begin{equation}\label{eq-UA_to_Heff}
        \mathcal{U}\mathcal{A}\mathcal{U}^{-1}=H_{\text{eff}}(H,d)=\sum_{j=1}^d E_j\ketbra{\psi_j}{\psi_j}.
    \end{equation}
\end{lemma}

\begin{proof}
    First of all, note that
    \begin{equation}\label{eq-U_proj}
        \mathcal{U}\ket{\alpha_j}=\ket{\psi_j}\quad\forall~j\in\{1,2,\dotsc,d\},
    \end{equation}
    which holds due to Lemma~\ref{lem-proj_vectors}. Next, using Eq.~\eqref{eq-proj_vectors_tilde}, we find that $\mathcal{U}=\sum_{j=1}^d\ketbra{\psi_j}{j}T^{-1}=ST^{-1}$, where $S\coloneqq\sum_{j=1}^{d}\ketbra{\psi_j}{j}$. Since the vectors $\ket{\psi_j}$ are linearly independent, the operator $S$ is invertible. Furthermore, because the vectors $\ket{\psi_j}$ are orthonormal, $S^{-1}=S^{\dagger}=\sum_{j=1}^d\ketbra{j}{\psi_j}$. Therefore, the inverse of $\mathcal{U}$ exists and is equal to $\mathcal{U}^{-1}=TS^{\dagger}$.
    Next, using Eq.~\eqref{eq:H_unpert}, we have that
    \begin{equation}
        P_0(\lambda V)=P_0(H-E_0^{(0)}P_0-Q)=P_0H-E_0^{(0)}P_0,
    \end{equation}
    so that
    \begin{equation}
        \mathcal{A}=P_0H\mathcal{U}-E_0^{(0)}P_0\mathcal{U}.
    \end{equation}
    However, using Eq.~\eqref{eq-proj_E0}, we find that $P_0\mathcal{U}=P_0$, which means that
    \begin{equation}
        \mathcal{A}=P_0H\mathcal{U}-E_0^{(0)}P_0.
    \end{equation}
    For this reason, because $\mathcal{U}\mathcal{U}^{-1}$ is equal to the projection onto the subspace $\mathcal{E}$, we obtain
    \begin{align}
        \mathcal{U}\mathcal{A}\mathcal{U}^{-1}&=\mathcal{U}P_0H\mathcal{U}\mathcal{U}^{-1}-E_0^{(0)}\mathcal{U}P_0\mathcal{U}^{-1}\\
        \nonumber
        &=\mathcal{U}P_0\sum_{j=1}^d\widetilde{E}_j\ketbra{\psi_j}{\psi_j}-E_0^{(0)}\mathcal{U}P_0\mathcal{U}^{-1}\\
        &=\sum_{j=1}^d\widetilde{E}_j\ketbra{\psi_j}{\psi_j}-E_0^{(0)}\mathcal{U}P_0\mathcal{U}^{-1}.
         \nonumber
    \end{align}
    Finally, using Eq.~\eqref{eq-proj_E0}, along with Eq.~\eqref{eq-U_proj} and the fact that $\bra{\widetilde{\alpha}_j}\mathcal{U}^{-1}=\bra{j}T^{-1}TS^{\dagger}=\bra{\psi_j}$, we find that
    \begin{equation}
        \mathcal{U}P_0\mathcal{U}^{-1}=\sum_{j=1}^d\mathcal{U}\ketbra{\alpha_j}{\widetilde{\alpha}_j}\mathcal{U}^{-1}=\sum_{j=1}^d\ketbra{\psi_j}{\psi_j}.
    \end{equation}
    Therefore,
    \begin{equation}
        \mathcal{U}\mathcal{A}\mathcal{U}^{-1}=\sum_{j=1}^d(\widetilde{E}_j-E_0^{(0)})\ketbra{\psi_j}{\psi_j}=\sum_{j=1}^d E_j\ketbra{\psi_j}{\psi_j},
    \end{equation}
    as required.
\end{proof}

\subsection{Perturbative expansion of \texorpdfstring{$\mathcal{U}$}{U}}
\label{A:expansion of U}

We now derive a perturbative expansion for $\mathcal{U}$, from which we obtain an expansion for $\mathcal{A}$, which will allow us to obtain a perturbative expansion for the effective Hamiltonian $H_{\text{eff}}(H,d)$ via Eq.~\eqref{eq-UA_to_Heff}. Crucial to obtaining the perturbative expansion of $\mathcal{U}$ is the following fact.

\begin{lemma}[Towards a perturbative expansion]
    The linear operator $\mathcal{U}$ defined in Eq.~\eqref{eq-pert_U} satisfies
    \begin{equation} \label{A-eq:governing-equation}
        \mathcal{U} = P_0 + A^{-1}Q_0 \lambda (V\mathcal{U} - \mathcal{U}V\mathcal{U}),
    \end{equation}
    where $Q_0\coloneqq\II -P_0$ and $A\coloneqq E^{(0)}_0\II - H^{(0)}$.
\end{lemma}

\begin{proof}
    From Eq.~\eqref{A-eq:intermediate}, we have that
\begin{equation} \label{A-eq:intermediate_2}
    \left( H-E^{(0)}_0\II \right) |\psi_j\rangle = E_j |\psi_j\rangle \quad\forall~j\in\{1,2,\dotsc,d\}.
\end{equation}
Multiplying both sides of this equation from the left by $P_0$ gives us
\begin{align}
    P_0(H-E_0^{(0)}\II)\ket{\psi_j}&=P_0E_j\ket{\psi_j},\\
    \Rightarrow P_0(H^{(0)}+\lambda V-E_0^{(0)}\II)\ket{\psi_j}&=E_j\ket{\alpha_j},\\
    \Rightarrow \lambda P_0V\ket{\psi_j}&=E_j\ket{\alpha_j},
\end{align}
for all $j\in\{1,2,\dotsc,d\}$, where to obtain the right-hand side of the second line we used Eq.~\eqref{eq-proj_vectors}. Now, multiplying both sides of the last line by $\mathcal{U}$ gives us
\begin{align}
    \lambda\mathcal{U}P_0V\ket{\psi_j}&=E_j\mathcal{U}\ket{\alpha_j},\\
    \Rightarrow \lambda\mathcal{U}V\ket{\psi_j}&=E_j\ket{\psi_j}, \label{A-eq:luv=e}
\end{align}
for all $j\in\{1,2,\dotsc,d\}$, where to obtain the left-hand side of the last line we used the fact that $\mathcal{U}P_0=\mathcal{U}$, which can be straightforwardly verified using Eq.~\eqref{eq-proj_E0}. For the right-hand side of the last line, we used Eq.~\eqref{eq-U_proj} the fact that $\mathcal{U}\ket{\alpha_j}=\ket{\psi_j}$ for all $j\in\{1,2,\dotsc,d\}$. Using Eq.~\eqref{A-eq:intermediate_2}, we find that Eq.~\eqref{A-eq:luv=e} can be written as
\begin{align}
    &\lambda\mathcal{U}V\ket{\psi_j}=(H^{(0)}-E_0^{(0)}\II)\ket{\psi_j},\\
    &\Rightarrow (H-E_0^{(0)}\II-\lambda\mathcal{U}V)\ket{\psi_j}=0,
\end{align}
for all $j\in\{1,2,\dotsc,d\}$. Multiplying the last line by $\bra{\widetilde{\alpha}_j}$ from the right leads to
\begin{equation}
    (H-E_0^{(0)}\II-\lambda\mathcal{U}V)\ketbra{\psi_j}{\widetilde{\alpha}_j}=0\quad\forall~j\in\{1,2,\dotsc,d\}.
\end{equation}
Summing over all $j\in\{1,2,\dotsc,d\}$ therefore leads to
\begin{equation}
    (H-E_0^{(0)}\II-\lambda\mathcal{U}V)\mathcal{U}=0,
\end{equation}
which is equivalent to
\begin{equation}
    \left( E^{(0)}_0\II - H^{(0)} \right) \mathcal{U} = \lambda V \mathcal{U} - \lambda \mathcal{U} V \mathcal{U}.
\end{equation}
Letting $Q_0\coloneqq\II-P_0$ be the projector onto the span of the excited states of $H^{(0)}$, we have
\begin{equation} \label{A-eq:U}
    \mathcal{U} = P_0 \mathcal{U} + Q_0 \mathcal{U},
\end{equation}
which means that 
\begin{align}
    \left( E^{(0)}_0\II - H^{(0)} \right) \mathcal{U} & = \left( E^{(0)}_0\II - H^{(0)} \right) (P_0 \mathcal{U} + Q_0 \mathcal{U})  \\ & 
    = 0 + \left( E^{(0)}_0\II - H^{(0)} \right) Q_0 \mathcal{U},
    \nonumber
\end{align}
leading to 
\begin{equation}
    \left( E^{(0)}_0\II-H^{(0)} \right) Q_0 \mathcal{U} = \lambda V\mathcal{U} - \lambda \mathcal{U}V\mathcal{U}.
\end{equation}
Next, using the fact that $P_0\mathcal{U}=P_0$, we obtain
\begin{equation}
\begin{split}
    Q_0 (V\mathcal{U} - \mathcal{U}V\mathcal{U}) & 
    = (\II - P_0) (V\mathcal{U} - \mathcal{U}V\mathcal{U}) \\ & 
    = V\mathcal{U} - \mathcal{U}V\mathcal{U} + P_0 (\mathcal{U}V\mathcal{U} - V\mathcal{U}) \\ & 
    = V\mathcal{U} - \mathcal{U}V\mathcal{U}.
\end{split}
\end{equation}
Then, because $H^{(0)} - E^{(0)}_0\II$ has a well defined inverse on $\mathcal{E}^{(0)\perp}$, we can write 
\begin{equation}
    Q_0 \mathcal{U} = \left(E^{(0)}_0\II - H^{(0)} \right)^{-1} Q_0 \lambda (V\mathcal{U} - \mathcal{U}V\mathcal{U}).
\end{equation}
Inserting this result into Eq.~\eqref{A-eq:U} yields 
\begin{equation}
    \mathcal{U} = P_0 \mathcal{U} +\left( E^{(0)}_0\II - H^{(0)} \right)^{-1} Q_0 \lambda (V\mathcal{U} - \mathcal{U}V\mathcal{U}),
\end{equation}
which simplifies to
\begin{equation} 
    \mathcal{U} = P_0 +\left(E^{(0)}_0\II -  H^{(0)}\right)^{-1} Q_0 \lambda (V\mathcal{U} - \mathcal{U}V\mathcal{U}),
\end{equation}
as required.
\end{proof}

We have, therefore, obtained the governing equation for $\mathcal{U}$. This equation is well suited for expansion in powers of $\lambda$ and can be expanded as 
\begin{equation}
    \mathcal{U} = \sum\limits_{m=0}^\infty \mathcal{U}^{(m)},
\end{equation}
with $\mathcal{U}^{(m)}$ being the $m^{\text{th}}$-order term.
Substituting $\mathcal{U}$ into Eq.~\eqref{A-eq:governing-equation} gives the recurrence relations
\begin{align}
    \mathcal{U}^{(0)} & = P_0 ,\\
    \mathcal{U}^{(m)} & = A^{-1} Q_0\lambda \left( V\mathcal{U}^{(m-1)} - \sum\limits_{p=1}^{m-1} \mathcal{U}^{(p)}V\mathcal{U}^{(m-p-1)}\right),
\end{align}
with $A = E^{(0)}_0\II - H^{(0)}$. 
\begin{figure}
    \centering
    \includegraphics{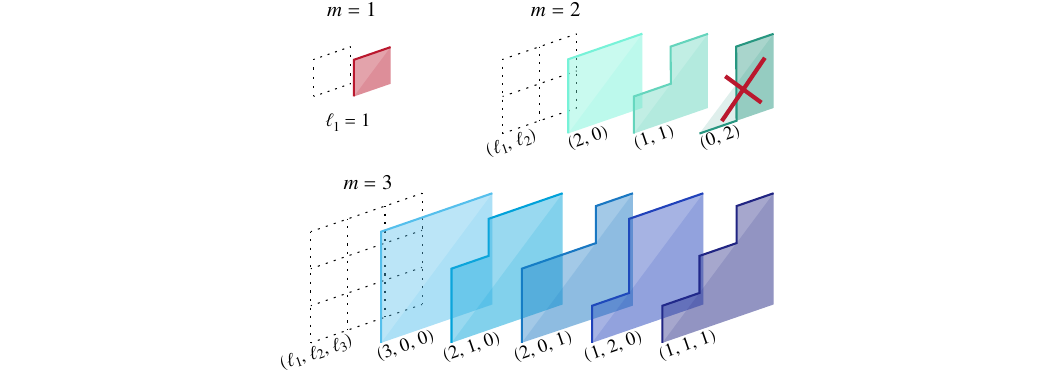}
    \caption{Visual representation of the allowed set of indices in the expansion of $\mathcal{U}$ as proposed by Bloch in Ref.~\cite{blochTheoriePerturbationsEtats1958}. All allowed diagrams are shown for $m\in\{1, 2, 3\}$ and additionally, one not contributing diagram for $m=2$ since it goes below the diagonal (it does not fulfill the third property in Eq.~\eqref{A-eq:U-indexes}).}
    \label{fig:staircase-diagrams}
\end{figure}
Let 
\begin{equation}\label{eq-S_ell}
    S^\ell \coloneqq \begin{cases}
        - P_0 \quad & \text{if } \ell=0, \\
        A^{-\ell}Q_0 & \text{if } \ell>0.
    \end{cases}
\end{equation} 
Then, 
\begin{equation} \label{Aeq:U-expansion}
    \mathcal{U}^{(m)} = \lambda^m {\sum}' S^{\ell_1} V S^{\ell_2} V \dots V S^{\ell_n} V P_0,
\end{equation}
where the sum is over all sets of indices $\{\ell_i\}_{i=1}^m$ that fulfill the following conditions~\cite{blochTheoriePerturbationsEtats1958}:
\begin{align} \label{A-eq:U-indexes}
    & \ell_i \geq 0 \quad \forall~i \in \{1, \dotsc, m \},\\
    & \ell_1 + \dots + \ell_{m-1} = m, \\
    & \ell_1 + \dots + \ell_{p} \geq p \quad \forall~p \in \{1, \dotsc, m-1 \}.
\end{align}

\begin{lemma}[Explicit form of terms]
    For $\ell>0$, the terms $S^\ell$ defined in Eq.~\eqref{eq-S_ell} are given by
    \begin{equation}
        S^{\ell}=\sum_{j\neq 0} (E_0^{(0)}-E_j^{(0)})^{-\ell}P_j\quad (\ell>0),
    \end{equation}
    where $P_j$ is the spectral projection of $H^{(0)}$ corresponding to the energy $E_j^{(0)}$.
\end{lemma}

\begin{proof}
To obtain this result, we use the fact that
\begin{equation}
    H^{(0)}=E_0^{(0)}P_0+\sum_{j\neq 0}E_j^{(0)}P_j,\quad P_0+\sum_{j\neq 0}P_j=\II.
\end{equation}
This implies that
\begin{equation}
    E_0^{(0)}\II-H^{(0)}=\sum_{j\neq 0}(E_0^{(0)}-E_j^{(0)})P_j,
\end{equation}
which in turn implies that
\begin{equation}
    A^{-\ell}=(E_0^{(0)}\II-H^{(0)})^{-\ell}=\sum_{j\neq 0}(E_0^{(0)}-E_j^{(0)})^{-\ell}P_j,
\end{equation}
and because this is supported entirely on the orthogonal complement of the ground space of $H^{(0)}$, we have that $A^{-\ell}Q_0=A^{-\ell}$, proving the desired result.
\end{proof}

From the expansion of $\mathcal{U}$, we can write the expansion of $\mathcal{A}$ as
\begin{equation}
    \mathcal{A} = \sum\limits_{m=1}^\infty \mathcal{A}^{(m)}
\end{equation}
with
\begin{equation} \label{Aeq:expansion-A}
    \mathcal{A}^{(m)} = \lambda^m {\sum}''P_0 V S^{\ell_1} V S^{\ell_2} V \dots V S^{\ell_{m-1}} V P_0,
\end{equation}
with the sum over all sets of $m-1$ indices $\{\ell_i\}_{i=1}^{m-1}$ such that
\begin{align} \label{Aeq:Bloch-A-indexes}
    & \ell_i \geq 0 \quad \forall~i\in\{1, 2, \ldots, m-1 \}, \\
    & \ell_1 + \dots + \ell_{m-1} = m-1 \geq 0, \\
    & \ell_1 + \dots + \ell_{p} \geq p \quad \forall~p \in \{1, \dotsc, m-2 \}.
\end{align}

Finally, using Eq.~\eqref{eq-UA_to_Heff}, the perturbative expansion of $H_{\text{eff}}(H,d)$, up to $m^{\text{th}}$ order in perturbation theory, is
\begin{equation}\label{Aeq:def-Heff}
    H_{\text{eff}}(H,d)=\mathcal{U}\mathcal{A}^{(\leq m)}\mathcal{U}^{-1}+\mathcal{O}(\lambda^{m+1}),\quad \mathcal{A}^{(\leq m)}\coloneqq\sum_{j=1}^m \mathcal{A}^{(j)}.
\end{equation}
As shown in Ref.~\cite{blochTheoriePerturbationsEtats1958} (see also Ref.~\cite[Appendix~B]{jordanPerturbativeGadgetsArbitrary2008a}), the following condition is sufficient to guarantee convergence of this perturbative expansion:
\begin{equation} 
\label{A-eq:convergence-condition}
    \| \lambda V \| < \frac{\gamma}{4},
\end{equation}
where $\gamma$ is the gap between $ E_0^{(0)} $ and the energy of the first excited state of $ H^{(0)}$. We note that this bound is not tight~\cite{caoEfficientOptimizationPerturbative2017} and that significantly larger $\lambda$ might still lead to similar results, even outside of the regime in which the perturbative expansion converges~\cite{bravyiQuantumSimulationManyBody2008,bauschGadgets2020}.

\section{Proof of Theorem \ref{theorem:main-result}} 
\label{A:proof}

In this section, we show how to get from Definition~\ref{def:hamiltonians} to Theorem~\ref{theorem:main-result} using the expansion presented in the previous section. 
In particular, we analyze what the obtained formulas in Appendix~\ref{A:Bloch-expansion} imply for the construction of the novel perturbative gadget introduced in Definition~\ref{def:hamiltonians}, and thereby prove Theorem~\ref{theorem:main-result}. 
The proof also serves as the groundwork for the extensions presented in Appendix~\ref{A:extensions} and the generalization in Section~\ref{sec:recipe}. 

For convenience, let us recall the definition of our gadget Hamiltonian from Definition~\ref{def:hamiltonians}. Given a $k$-body target Hamiltonian $H^{\text{target}}$ acting on a Hilbert space $\mathcal{H}^{\text{target}}$ with $\operatorname{dim}(\mathcal{H}^{\text{target}})=2^n$, the associated gadget Hamiltonian is
\begin{equation}
    H^{\text{gad}} \coloneqq \sum\limits_{ s = 1}^r H^\text{aux}_s + \lambda \sum\limits_{ s = 1}^r V_s,
\end{equation}
where
\begin{align} 
    H^\text{aux}_s &\coloneqq \sum\limits_{ j = 1}^k \frac{1}{2} \left( \II_{s,j} - Z_{s,j}^{\text{aux}} \right)
    = \sum\limits_{ j = 1}^k \ketbra{1}{1}_{s,j}^{\text{aux}}, \\
    V_s &\coloneqq \sum\limits_{j=1}^k \tilde{c}_{s,j} \sigma_{s,j}^\text{target} \otimes X_{s,j}^\text{aux} \otimes X_{s,(j+1)\text{ mod }k}^\text{aux} \label{eq:def-V_2}
\end{align}
and $\tilde{c}_{s,j}=-(-1)^k c_s$ if $j=1$ or $\tilde{c}_{s,j}=1$ otherwise.
We consider $\sigma_{s,j} \in \left\{ \hat{\boldsymbol{n}} \cdot \boldsymbol{\sigma}:\hat{\boldsymbol{n}} \in \mathds{R}^3,\, \hat{\boldsymbol{n}}\cdot\hat{\boldsymbol{n}}=1,\, \boldsymbol{\sigma}= (X, Y, Z) \right\}$.
We refer to the Hilbert spaces of each of the auxiliary registers as $\mathcal{H}^\text{aux}_s$, where $\operatorname{dim}(\mathcal{H}^\text{aux}_s) = 2^k$. 
The complete Hilbert space acted upon by our gadget Hamiltonian is thus
\begin{equation}
    \mathcal{H}^\text{total} \simeq \mathcal{H}^\text{target} \otimes \mathcal{H}^\text{aux}_1 \otimes \ldots \otimes \mathcal{H}^\text{aux}_r.
\end{equation}
Our goal in this section is to use the perturbative expansion of the previous section, by taking $H^{(0)}\equiv \sum_{s=1}^r H_s^{\text{aux}}$ and $V\equiv \sum_{s=1}^r V_s$, and showing that 
\begin{equation}
    H_{\text{eff}}(H^{\text{gad}},2^n) = a H^\text{target}\otimes\ketbra{\phi_0}{\phi_0}  + b\Pi + O_\text{err},
\end{equation}
where $\Pi$ is the projector onto the low-energy subspace of $H^{\text{gad}}$, $\ket{\phi_0}\in\mathcal{H}^{\text{aux}}$ is a state vector, $a$ is a scaling factor, and $b$ is a shift on the whole subspace of interest. We remark that the above equation is in accordance with \cite[Definition~11]{bauschGadgets2020}, which provides a general definition of what it means for one Hamiltonian to mimic another Hamiltonian in its low-energy subspace.

\begin{figure}
    \centering
    \includegraphics{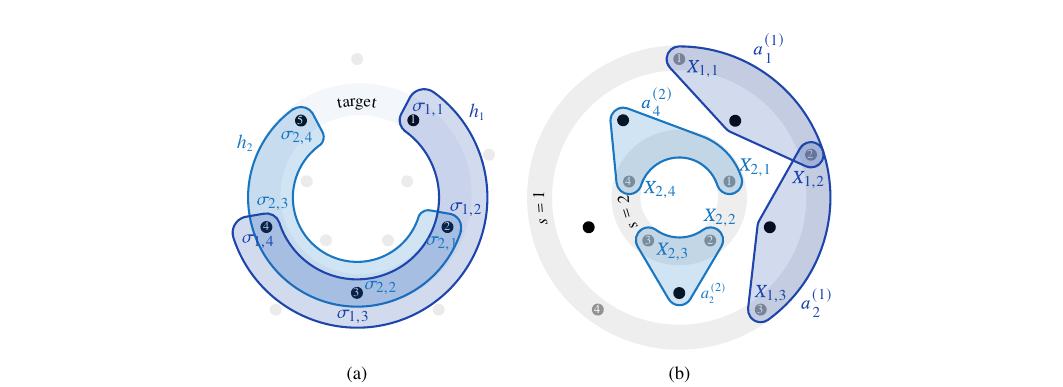}
    \caption[Registers of the gadget Hamiltonian]{
    Illustration of the different relevant qubit registers in the construction of the gadget Hamiltonian for the case of $n=5$, $k=4$, and $r=2$. 
    Shown in (a) are the two target Hamiltonian terms $h_1$ and $h_2$ on the five-qubit target register with the corresponding single-qubit operators $\sigma_{s,j}$. 
    (b) shows the two four-qubit auxiliary registers, and two of the $a^{(s)}_j$ operators displayed on each.
    As displayed for $\sigma_{2,3}$, the indices $(s,j)$ for operators on the target register should always be understood by referring to $h_s$. Indeed, $\sigma_{2,3}$ acts on qubit $4$. On the other hand, the indices for operators on the auxiliary qubits are more straightforward: $X_{2,3}$ is a Pauli $X$ on the third qubit of the second register.}
    \label{fig:registers}
\end{figure}

\subsection{Useful properties}
\label{A:useful_properties}
For illustrative purposes, and to help in some coming steps of the derivation, let us write down the lowest-order terms of the expansion of $\mathcal{A}$ from Eq.~\eqref{Aeq:expansion-A} explicitly. This is \begin{eqnarray}
    \mathcal{A}^{(1)} &=& \lambda P_0 V P_0, \\
    \mathcal{A}^{(2)} & =& \lambda^2 P_0 V S^1 V P_0 , \\
    \mathcal{A}^{(3)} & =& \lambda^3 P_0 V S^1 V S^1 V P_0 + \lambda^3 P_0 V S^2 V S^0 V P_0.
\end{eqnarray}
Let us consider a few properties of the construction defined in Definition~\ref{def:hamiltonians}, which we use below. 
First of all, note that the gadget Hamiltonian is three-body by construction.
Then, the unperturbed Hamiltonian acting solely on the auxiliary registers can be rewritten as
\begin{equation}
    H^{\text{aux}} = \sum\limits_{s=1}^r \sum\limits_{j=1}^k |1\rangle \! \<1|_{s,j}
\end{equation} 
and its ground state is the computational basis state vector $|0\rangle^{\otimes rk}$ with corresponding eigenvalue $0$. 
Consequently, the ground space of the unperturbed Hamiltonian is given by
\begin{equation}
    \mathcal{E}^{(0)} = \textrm{span} \left\{ |\varphi\rangle \otimes |0\rangle^{\otimes rk}:  |\varphi\rangle \in \mathcal{H}^\text{target}\right\},
\end{equation}
with the corresponding projector
\begin{equation}
    P_0 = \II^{\otimes n} \otimes \left( |0\rangle \! \<0| \right)^{\otimes rk},
\end{equation}
whose support is $2^n$-dimensional, as for the original target Hamiltonian.
Furthermore, the energy of any state in the computational basis is given by the Hamming weight of the state with respect to the auxiliary register.
Lastly, let us define 
\begin{equation}
    a^{(s)}_j \coloneqq X_{s,j}^{\text{aux}} \otimes X_{s,(j+1)\text{ mod }k}^{\text{aux}}
\end{equation}
as the parts of the perturbation from Eq.~\eqref{eq:def-V_2} acting on the auxiliary registers.
Then, these operators fulfill the useful relations
\begin{equation} 
    \left[ a^{(s)}_i, a^{(p)}_j \right] = 0 \quad \forall~i,j,s,p,\qquad  \prod_{j=1}^k a^{(s)}_j = \bigotimes_{j=1}^k X_{s,j}^2 = \II^{\otimes k} \quad \forall~s. \label{Aeq:aj-properties}
\end{equation}
The first is a direct consequence of having only Pauli-$X$ operators in the construction of $a^{(s)}_j$. The second is due to the fact that the operators $a_j^{(s)}$ are constructed in a cyclic manner and the fact that the Pauli-$X$ operator squares to the identity.

\subsection{Simplified target Hamiltonian}
\label{A:simplified_target_H}
Now, let us first consider the simplified example of a target Hamiltonian comprising only a single term with unit norm, that is
\begin{equation}
    H^\text{target} = \sigma_{1} \sigma_{2} \ldots \sigma_{k}.
\end{equation}
For this case, we can omit the subscript $s$. The corresponding gadget Hamiltonian is then given by
\begin{equation}
    H^\text{gad} = H^\text{aux} + \lambda V
\end{equation}
with
\begin{equation}
    H^\text{aux} = \sum\limits_{j=1}^k \frac{1}{2}(\II - Z_{j}) 
\end{equation}
and
\begin{equation}
    V = \sum\limits_{j=1}^k \left(-(-1)^k\right)^{\delta_{1j}} \sigma_{j} \otimes X_{j} \otimes X_{(j+1)\text{ mod }k}.
\end{equation}
The $2^{k}$-dimensional unperturbed ground space is 
\begin{equation}
    \mathcal{E}^{(0)} = \textrm{span} \left\{ |\varphi\rangle^\text{target} \otimes | 0 \rangle^{\otimes k}: |\varphi\rangle \in \mathcal{H}^\text{target}\right\}.
\end{equation}
The projector on the unperturbed ground space can consequently be written as 
\begin{equation} \label{eq:gadget-ground-projector}
    P_0 = \II^{\otimes k} \otimes \left( |0\rangle \! \<0| \right)^{\otimes k}.
\end{equation}
First, let us look at the auxiliary part of the Hamiltonian to understand its effect.
It can be interpreted as a penalization on flipped qubits: each term has a $0$ contribution if the affected qubit is in the $|0\rangle$ state but has a penalty of $1$ if the qubit is in $|1\rangle$.
The gap is then $\gamma = 1$ and taking into account that $\|V\| = k$, the convergence of the expansion in Eq.~\eqref{Aeq:U-expansion} is guaranteed for 
$\lambda \leq \frac{1}{4k}$. 

Let us now study the first terms in the expansion of $\mathcal{A}$ from Eq.~\eqref{Aeq:expansion-A}, i.e.,
\begin{equation}
    \mathcal{A}^{(\leq 2)} = \lambda P_0 V P_0 + \lambda^2 P_0 V S^1 V P_0.
\end{equation}
Every term of this expansion is sandwiched between two projectors $P_0$. 
The only terms with a non-trivial contribution are then those that take a state from $ \mathcal{E}^{(0)} $ and return it to $ \mathcal{E}^{(0)}$, i.e., those that leave the auxiliary register in the all-zero state.
That is not the case for $ \mathcal{A}^{(1)} $ in which the application of a single term of the perturbation $V$ necessarily kicks the state out of the ground space by flipping two auxiliary qubits. Consequently, $\mathcal{A}^{(1)}=0$.
For $ \mathcal{A}^{(2)} $, each term of the first $V$ excites the system to a state with two flipped auxiliary qubits.
The only contribution from $S^1$ is then due to the projector on the second excited subspace $P_2$ with corresponding eigenenergy $ E^{(0)}_2 = 2$.
This gives the denominator in $S^1$, resulting in
\begin{equation}
    \mathcal{A}^{(\leq 2)} = -\frac{\lambda^2}{2} P_0 V^2 P_0.
\end{equation}
Now, let us examine $V^2$. This two-fold application of the perturbation $V$ has to act on the same two qubits, flipping them back to their original state. 
All cross-terms in
\begin{equation} \begin{split} \label{Aeq:cross-terms}
    V^2 &
    = \sum\limits_{i,j=1}^k \sigma_{i} \sigma_{j} \otimes a_{i} a_{j} \\ &
    = \sum\limits_{i,j=1}^k \sigma_{i} \sigma_{j} \otimes X_{i} X_{(i+1)\text{ mod }k} X_{j} X_{(j+1)\text{ mod }k},
\end{split} \end{equation}
with $i\neq j$, do not contribute since some of the Pauli-$X$ operators on auxiliary qubits survive and leave the state outside of $ \mathcal{E}^{(0)}$. 
For visualization, we refer to Fig.~\ref{fig:gadget-analysis}(c).
Additionally, for $i = j$, and even on the target register, the terms that do contribute are necessarily squared operators and thus only identities:
\begin{equation}
    \sum\limits_{j=1}^k \sigma_{j}^2 \otimes \left(X_{j} X_{(j+1)\text{ mod }k} \right)^2 
    = \sum\limits_{j=1}^k \II=k\II.
\end{equation}
Altogether, we have that 
\begin{equation}
    \mathcal{A}^{(\leq 2)} = -\frac{\lambda^2 k}{2} P_0.
\end{equation}

\begin{figure*}[t]
    \centering
    \includegraphics[width=.99\linewidth]{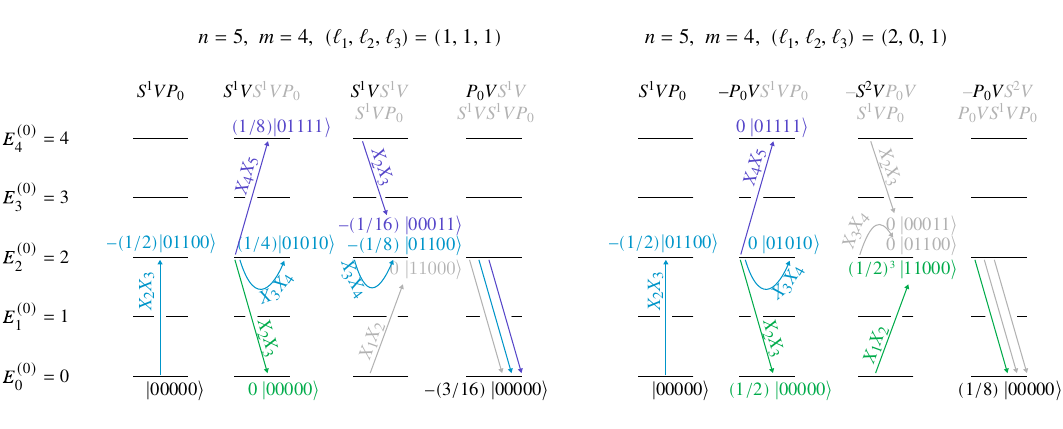}
    \caption[Computation of $\Xi$]{Illustration of the interplay of $V$ and $S^\ell$ in $\mathcal{A}$. Shown are
    some possible contributions to $\alpha_4$ in the case of seven qubits; see Eq.~\eqref{Aeq:A-lower-k}.
    Displayed are two choices of indices $\ell_i$ fulfilling~\eqref{Aeq:Bloch-A-indexes}
    and some terms from $V^4$ that starting from a state in $\mathcal{E}^{(0)}$ return to $\mathcal{E}^{(0)}$ for each choice.
    At each step of application of $S^{\ell_i}V$, the state is changed by $V$ by flipping two qubits, and $S^{\ell_i}$ weights the resulting state by a power of the inverse of the energy or $0$.}
    \label{fig:gadget-perturbations}
\end{figure*}

Similarly, for the next orders in the perturbative expansion, each pair of flipped auxiliary qubits has to be flipped back by the same operator in order to result in a non-zero contribution. 
For this to be true, $m$ is only allowed to take even values, and the perturbation is always applied as $P_0 V^{m} P_0 \propto P_0 $. 
That means that all contributions below order $k$ are proportional to $P_0$ and thus
\begin{equation} \label{Aeq:A-lower-k}
    \mathcal{A}^{(\leq k-1)} = \left( \sum\limits_{m=0}^{k-1} \alpha_m \lambda^m \right) P_0.
\end{equation}
Here, $\alpha_m$ is a coefficient depending on the number of combinations of $m$ excitations returning a state to $ \mathcal{E}^{(0)} $, and the corresponding energy penalties picked up along the way, which is independent of $\lambda$. 
To strengthen the intuition on the construction of $\alpha_m$ we refer to Fig.~\ref{fig:gadget-perturbations}, which displays a few examples of terms that do and do not contribute.

A different behavior can be first observed at $k^{\text{th}}$ order in perturbation theory.
There, by Eq.~\eqref{Aeq:aj-properties}, we can construct terms acting $k$ times on the auxiliary register as 
$ \prod_{j=1}^k a_j = \II^{\otimes k} $.
In other words, all qubits from the auxiliary register are flipped twice and thus returned to $ \mathcal{E}^{(0)} $, while each qubit in the target register is acted upon only once by the corresponding $\sigma_{j}$.
This results in all elements from $H^\text{target}$ being applied to the target register, yielding
\begin{equation}
    \mathcal{A}^{(k)}=-(-\lambda)^k \sum_{\boldsymbol{\ell}} \frac{1}{\xi_{\boldsymbol{\ell}}} P_0 \left( -(-1)^k \sigma_{1} \sigma_{2} \dots \sigma_{k} \otimes \II^{\otimes k}\right) P_0 = \lambda^k \sum_{\boldsymbol{\ell}} \frac{1}{\xi_{\boldsymbol{\ell}}} P_0 \left( H^\text{target} \otimes \II^{\otimes k} \right) P_0.
 \end{equation}
The sum over $\boldsymbol{\ell}$ is over all possible permutations in the application order of the $ \left( \sigma_{j},\ a_j \right)$ pairs, while $\xi_{\boldsymbol{\ell}}= \prod_{i=1}^k E_{\ell_i}$ is the factor originating from the corresponding energy penalties $E_{\ell_i} = \| H^{\text{aux}} a_{\ell_i} \ldots a_{\ell_0} |0\rangle^{\otimes k} \|_2$.
Defining 
\begin{equation}
{\Xi}^{-1} \coloneqq \sum_{\boldsymbol{\ell}} \frac{1}{\xi_{\boldsymbol{\ell}}},
\end{equation}
the full expansion of $ \mathcal{A} $ up to $k^{\text{th}}$ order in perturbation theory considering all possible combinations can then be written as
\begin{equation} \begin{split}
    \mathcal{A}^{(\leq k)} & 
    = f(\lambda) P_0 + \frac{\lambda^k}{\Xi} P_0 \left( H^\text{target} \otimes \II^{\otimes k} \right) P_0 \\ &
    = f(\lambda) P_0 + \frac{\lambda^k}{\Xi}  H^\text{target} \otimes (|0\rangle \! \<0|)^{\otimes k}.
\end{split} \end{equation}
Applying these results to Eq.~\eqref{Aeq:def-Heff} yields
\begin{equation} \label{eq:Jordan-Heff-gadget} 
\begin{split}
    H_\text{eff}( H^\text{gad}, 2^{n}) 
    &= \mathcal{U}f(\lambda) P_0 \mathcal{U}^{-1} 
    + \mathcal{U}\left[\frac{\lambda^k}{\Xi} H^\text{target} \otimes (|0\rangle \! \<0|)^{\otimes k} + \mathcal{O}(\lambda^{k+1}) \right] \mathcal{U}^{-1} \\ 
    &= f(\lambda) \Pi 
    + \frac{\lambda^k}{\Xi} H^\text{target} \otimes (|0\rangle \! \<0|)^{\otimes k} + \mathcal{O}(\lambda^{k+1}),
\end{split}
 \end{equation}
where $\mathcal{U} P_0 \mathcal{U}^{-1} = \Pi$ is the projector onto the support of $H_\text{eff}(H^{\text{gad}},2^n)$. 
Similarly to Ref.~\cite{jordanPerturbativeGadgetsArbitrary2008a}, we used the fact that the operators $\mathcal{U}$ on both sides leave the second term unaffected up to errors of order $\mathcal{O}(\lambda^{k+1})$.

\subsection{General case}
\label{A:general_case}
For the general case stated in Definition~\ref{def:hamiltonians}, the argument is similar, with the exception that one has to bear in mind that there are now $r$ auxiliary registers of $k$ qubits each.
There are then more cross-terms in the powers of $V$ and Eq.~\eqref{Aeq:cross-terms} becomes
\begin{equation}
    V^2 = \sum\limits_{s,p=1}^r \sum\limits_{i,j=1}^k \tilde{c}_{s,i} c_{p,j} \sigma_{s,i} \sigma_{p,j} \otimes X_{s,i} X_{s,i+1} X_{p,j} X_{p,j+1}.
\end{equation}
Similarly, only terms with $s = p$ and $i = j$ contribute. 
More contributions appear from all possible combinations of the different terms in the powers of $V$, but in the end, Eq.~\eqref{Aeq:A-lower-k} still holds: 
All non-zero contributions to $\mathcal{A}$ at orders lower than $k$ are proportional to $P_0$.

At $k^{\text{th}}$ order in perturbation theory, the non-trivial terms appear when acting $k$ times on a single auxiliary register, since  Eq.~\eqref{Aeq:aj-properties} only holds when all operators operate on the same register $s$.

On the target register, it produces a contribution of the form $\prod_{j=1}^k \tilde{c}_{s,j} \sigma_{s,1} \sigma_{s,2} \dots \sigma_{s,k} = -(-1)^k c_{s} h_s$.
Considering all possible cross terms of this kind emerging from $V^k$ leads to similar terms for each auxiliary register, yielding
\begin{equation} \begin{split}
\label{a-eq:oder-k}
    \mathcal{A}^{(\leq k)} & 
    = f(\lambda) P_0 + \sum_{s=1}^r \frac{\lambda^k}{\Xi} P_0 \left( c_s h_s \otimes \II^{\otimes k}_s \right) P_0 \\ &
    = f(\lambda) P_0 + \frac{\lambda^k}{\Xi} P_0 \left( \sum_{s=1}^r c_s h_s \otimes \II^{\otimes k}_s \right) P_0 \\ &
    = f(\lambda) P_0 + \frac{\lambda^k}{\Xi} P_0 \left( H^\text{target} \otimes \II^{\otimes rk} \right) P_0.
\end{split} \end{equation}
Applying these results to Eq.~\eqref{Aeq:def-Heff} results in
\begin{equation}
\begin{split}
    H_\text{eff}( H^\text{gad}, 2^{n})
    &=\mathcal{U}f(\lambda) P_0 \mathcal{U}^{-1}+  \mathcal{U}\left[\frac{\lambda^k}{\Xi} P_0 \left( H^\text{target} \otimes \II^{\otimes rk} \right) P_0 + \mathcal{O}(\lambda^{k+1}) \right] \mathcal{U}^{-1} \\
    &= f(\lambda) \Pi 
    + \frac{\lambda^k}{\Xi} H^\text{target} \otimes (|0\rangle \! \<0|)^{\otimes rk} + \mathcal{O}(\lambda^{k+1}),
\end{split}
\end{equation}
which is what we have been aiming for. 
Furthermore, by Eq.~\eqref{A-eq:convergence-condition}, we need to upper bound the perturbation strength $\lambda$ by $\gamma/(4\lVert V \rVert)$. For the presented gadget we have $\gamma=1$ and $\lVert V\rVert\leq \left(\sum_{s=1}^r|c_s|+r(k-1)\right)$ by the triangle inequality and the fact that the operator norm of Pauli operators is equal to one. Although not tight, we can thus upper bound $\lambda$ by
\begin{equation}
    \lambda \leq \frac{1}{4 \left(\sum_{s=1}^r|c_s|+r(k-1)\right)},
\end{equation}
concluding the proof of Theorem~\ref{theorem:main-result}.

\section{Proof of Corollary \ref{corollary:minimas}}
\label{A:proof-corr}

Corollary~\ref{corollary:minimas} results from a direct application of Eq.~\eqref{eq:Heff} for sufficiently small $\lambda$. 
Since $f(\lambda)\Pi$ is a constant shift on the whole subspace of interest, the eigenstates of $H_\text{eff} \left(H^{\text{gad}}, 2^{n} \right) $ and $H_\text{eff} \left(H^{\text{gad}}, 2^{n} \right) - f(\lambda)\Pi$ are identical, and all energy gaps preserved.
We can then rewrite Eq.~\eqref{eq:Heff} as
\begin{equation} \label{eq:Heff-shifted}
    \widetilde{H}_{\text{eff}}(H^{\text{gad}}, 2^{n}, f(\lambda)) = \frac{\lambda^k}{\Xi} H^\text{target} \otimes \left(|0 \rangle\!\< 0 |\right)^{\otimes rk} + \lambda^{k+1} O^\text{err} ,
\end{equation}
with $ \widetilde{H}_{\text{eff}} \left(H^{\text{gad}}, 2^{n}, f(\lambda) \right)\coloneqq H_\text{eff} \left(H^{\text{gad}}, 2^{n} \right) - f(\lambda)\Pi $ and  $ \| O^\text{err} \| \in \mathcal{O}(1)$.
Since the error shifts any eigenvalue by at most $\| \lambda^{k+1} O^\text{err} \|$, choosing $\lambda$ such that 
\begin{equation}
\| \lambda^{k+1} O^\text{err} \| \leq  {\lambda^k} ( E^\text{target}_1 - E^\text{target}_0 )/{\Xi},
\end{equation}
where $E_0^{\text{target}}$ and $E_1^{\text{target}}$ are the ground and first-excited energies, respectively, of $H^{\text{target}}$, ensures that the ground space of the right-hand side of Eq.~\eqref{eq:Heff-shifted} remains separated from the first excited subspace.
Considering that the spectrum of $H^\text{target}$ is independent of $\lambda$, we are always able to find a sufficiently small $\lambda$, whose upper bound is given by
\begin{equation}
    \lambda \leq \frac{E^\text{target}_1 - E^\text{target}_0 }{\Xi \|O^\text{err} \|},
\end{equation}
concluding the proof.

\section{Extensions of the locality gadget construction} 
\label{A:extensions}

\subsection{Extension to mixed Pauli weights}
\label{A:extension1}

The gadget presented in Definition~\ref{def:hamiltonians} seems to be restricted to the case in which all terms $h_s$ of the target Hamiltonian act non-trivially on the same number of qubits, namely, that they all act non-trivially on exactly $k$ qubits. 
Here, we want to argue that such a restriction, although useful to simplify the derivations and formulas, is not necessary.
We can also consider the case of a Hamiltonian $H^\text{target}$ whose terms $h_s$ are not all acting non-trivially on the same number of qubits, i.e., 
\begin{equation} \label{Aeq:def-Hcomp}
    H^\text{target} = \sum\limits_{s=1}^r c_s h_s \quad \text{with} \quad h_s = \bigotimes\limits_{j=1}^{k_s}\sigma_{s,j}
\end{equation}
where $k_s$ and $k_{s'}$ may be different. Let us, therefore, define $k = \max_s k_s$. This means that all terms $h_s$ act non-trivially on at most $k$ qubits.

First of all, one should note that no step of the derivation in Appendix~\ref{A:proof} actually relies on the fact that 
$\sigma_{s,j}$ is a linear combination of the Pauli operators $X, Y, Z$.
What has been leveraged, on the other hand, is that without loss of generality and to simplify the derivation, we can assume the property $\sigma_{s,j}^2 = \II$.
In other words, nothing prohibits us from extending this definition to include identity terms, such that
$\sigma_{s,j} \in \left\{ \hat{\boldsymbol{n}} \cdot \boldsymbol{\sigma}:\hat{\boldsymbol{n}} \in \mathds{R}^3,\, \hat{\boldsymbol{n}}\cdot\hat{\boldsymbol{n}}=1,\, \boldsymbol{\sigma}= (X, Y, Z) \right\} \bigcup \{ \II \}$. We can therefore redefine the gadget from Eq.~\eqref{eq:def-Hgad} such that $H^\text{aux}_s$ stays as in Definition~\ref{def:hamiltonians}, but the 
perturbation changes to 
\begin{equation}
    V_s = \sum\limits_{j=1}^{k_s} \tilde{c}_{s,j} \sigma_{s,j}^\text{target} \otimes X_{s,j}^{\text{aux}} \otimes X_{s,(j+1)\text{ mod }k}^{\text{aux}} 
    + \sum\limits_{j=k_s+1}^k X_{s,j}^{\text{aux}} \otimes X_{s,(j+1)\text{ mod }k}^{\text{aux}}.\end{equation}
With this change, even terms with $k_s < k$ will only contribute at $k^{\text{th}}$ order in perturbation theory, and we recover our main result even for Hamiltonians of the form presented in Eq.~\eqref{Aeq:def-Hcomp}.

\subsection{Extension to arbitrary target locality}
\label{A:extension2}

Having extended the use of our gadget to target Hamiltonians whose terms act non-trivially on different numbers of qubits, we can also extend it for arbitrary target locality. 
For now, let us assume again that all terms $h_s$ act non-trivially on exactly $k$ qubits.
Reducing the locality to a value lower than three is out of reach for our gadget construction. However, we can construct gadget Hamiltonians with localities between three and $k$ that lead to the same overall result as Theorem~\ref{theorem:main-result}. 
Doing so might seem to be a counter-productive goal at first, but such constructions will require fewer additional qubits and might thus be useful when it comes to practical implementations. 
After all, the three-body perturbative gadget introduced here requires $rk$ additional qubits.
We start by discussing the simpler case in which $k$ is divisible by $k'-2$, where $k'$ is the target locality, and then extend the findings to the general case.
 
If, for instance, one chooses the target locality to be four instead of three, Eq.~\eqref{eq:def-V_2} can be changed to
\begin{equation} 
    V_s = \sum\limits_{j=1}^{k/2} \tilde{c}_{s,j} \sigma_{s,2j-1}^\text{target} \otimes \sigma_{s,2j}^\text{target} \otimes X_{s,j}^{\text{aux}} \otimes X_{s,(j+1)\text{ mod }k/2}^{\text{aux}}.
\end{equation}
Each perturbation term now acts on two target qubits instead of one, halving the order in perturbation theory at which the target Hamiltonian is recovered and would consequently only require $r\frac{k}{2}$ auxiliary qubits to obtain a result equivalent to Eq.~\eqref{eq:Heff}.

In general, we can construct a gadget with an arbitrary target locality between three and $k$ by acting on more target qubits at once in each term of the perturbation.
If the original target Hamiltonian is $k$-body and the target locality of $H^\text{gad}$ is $k'$, the resulting gadget construction requires $\frac{k}{k'-2}$ auxiliary qubits for each term of the target Hamiltonian. 
Two of the $k'$ operators of the tensor product have to be the pair of Pauli-$X$ operators on the auxiliary register and the rest can be populated with the corresponding Pauli operators from $h_s$.

The statements above assume that $\frac{k}{k'-2}$ is equal to an integer; now let us look at the case where this is not the case.
Let us define $k^\top = \lceil\frac{k}{k'-2} \rceil$ and  $k^\perp = \lfloor\frac{k}{k'-2} \rfloor$.
Using the insights from Appendix~\ref{A:extension1}, we can lift the divisibility requirement and can construct a $k'$-body gadget Hamiltonian with the perturbation given by
\begin{equation}
    V_s = \sum\limits_{j=1}^{k^\perp} \tilde{c}_{s,j} \bigotimes_{\ell=1}^{k'-2} \sigma_{s,(j-1)(k'-2)+\ell}^\text{target} \otimes X_{s,j}^{\text{aux}} \otimes X_{s,(j+1)\text{ mod } k^\top}^{\text{aux}}
    + \tilde{c}_{s,k^\perp+1} \bigotimes_{\mathclap{m=(k^\perp-1)(k'-2)+1}}^{k} \sigma_{s,m}^\text{target} \otimes X_{s,k^\perp}^{\text{aux}} \otimes X_{s,(k^\perp+1)\text{ mod }k^\top}^{\text{aux}}.
\end{equation}
The last term may act on less than $k'-2$ target qubits but the paired $X$ operator ensures that the only contribution is the one \emph{acting} $k$ times on the target register with $k_s$ Pauli operators and $k-k_s$ identities.

When $k^\perp = {k}/({k'-2})$, the final term does not contribute and we recover the divisible case. Otherwise, the last term implements the remaining operators of the target term $h_s$.
For a target Hamiltonian whose terms act non-trivially on $k$ qubits 
simultaneously and a targeted locality of the gadget Hamiltonian of $k'$, the number of required auxiliary qubits is $k^\top r$.

\section[Barren plateaus]{Can perturbative gadgets help address the barren plateau problem?}
\label{A:barren-plateaus}

In this section, we provide an in-depth study of our results when applying perturbative gadgets in the context of gate-based quantum computing, in particular in the context of variational quantum algorithms.
One might ask if and why one would even need a perturbative gadget without constraints of adiabaticity. 
After all, digital, gate-based quantum computing is not restricted by existing couplings due to the design of more intricate gate decompositions and the freedom of constructing arbitrary gates from universal gate sets. 
Although finding practical uses in general gate-based computations is as of yet an open problem, one fruitful application could be in the context of \emph{variational quantum algorithms (VQAs)}~\cite{cerezoVariationalQuantumAlgorithms2021a,bhartiNoisyIntermediatescaleQuantum2022,mccleanTheoryVariationalHybrid2016a}, which are based upon measurements of a Hamiltonian expanded in the Pauli basis and for which it was shown that the Hamiltonian's locality, both in the sense of geometric locality and few-body terms, can play a role in their performance~\cite{cerezoCostFunctionDependent2021a,uvarovBarrenPlateausCost2021a}.

VQAs rely upon \emph{parametrized quantum circuits} (PQCs) to evaluate a parameter-dependent cost function related to the expectation value of a set of observables and are greatly discussed in the setting of \emph{noisy intermediate scale quantum (NISQ)} devices~\cite{preskillQuantumComputingNISQ2018a}.
The probably simplest setting is that of the \emph{variational quantum eigensolver} (VQE)~\cite{peruzzoVariationalEigenvalueSolver2014a}, which refers to the case where the cost function is equal to the expectation value of a Hamiltonian $H^\text{target}$, i.e., 
\begin{equation} \label{eq:cost-observable}
    C^\text{target}(\boldsymbol{\theta}) = \Tr\!\left[ H^\text{target} U(\boldsymbol{\theta}) \rho U(\boldsymbol{\theta})^{\dagger} \right],
\end{equation}
for some initial state $\rho$, e.g., one given by the state vector $\ketbra{00\dots 0}{00\dots 0}$, and quantum circuit ansatz $\boldsymbol{\theta}\mapsto U(\boldsymbol{\theta})$.
Solving the problem corresponds to estimating the ground state energy and associated ground state vector of $H^\text{target}$ by classically optimizing the parameters $\boldsymbol{\theta}$ and employing a quantum device to evaluate the aforementioned cost function.

A significant obstacle in the successful optimization is the presence of so-called \emph{barren plateaus}, which reflect the phenomenon that the variance of gradients of the cost function with respect to the variational parameters decays exponentially in the number of qubits of the PQC.
First discussed for the hardware-efficient ansatz with a polynomial number of layers in Ref.~\cite{mccleanBarrenPlateausQuantum2018a}, the emergence of barren plateaus has been observed in other settings as well \cite{wangNoiseInducedBarrenPlateaus2021,marreroEntanglementInducedBarren2021, holmesConnectingAnsatzExpressibility2022, martinBP22} and hinders the optimization procedure due to the therefore exponential number of required measurement samples. Scalable VQAs thus rely on practical mitigation strategies.

Existing approaches build upon improving initialization strategies \cite{grantInitializationStrategyAddressing2019}, using correlated gates or restricted parameter spaces~\cite{holmesConnectingAnsatzExpressibility2022,volkoffLargeGradientsCorrelation2021}, intermediate measurements~\cite{wiersemaMeasurementinducedEntanglementPhase2021}, the transferability of smooth solutions~\cite{meleAvoidingBarrenPlateaus2022}, layerwise learning~\cite{skolikLayerwiseLearningQuantum2021}, pre-optimization~\cite{kieferovaQuantumGenerativeTraining2021,dborinMatrixProductState2021}, different types of ansätze~\cite{liuMitigatingBarrenPlateaus2022}, classical shadows~\cite{sackAvoidingBarrenPlateaus2022a}, or alternative loss landscapes \cite{kianiLearningQuantumData2022}.

Furthermore, in Refs.~\cite{Khatri2019,laroseVQSD,prietoVQLS}, it has surprisingly been found that for large system sizes and circuit depths a \emph{local} cost function, i.e., a cost function defined by a few-body Hamiltonian, results in more effective optimization compared to a \emph{global} cost function, i.e., a cost function defined by a many-body Hamiltonian, and that global cost functions exhibit vanishing gradients for large system sizes already at constant circuit depths.
The concept of such \emph{cost-function-dependent} barren plateaus was formalized in Refs.~\cite{cerezoCostFunctionDependent2021a,uvarovBarrenPlateausCost2021a}, in which it has been further shown that local cost functions can still be optimized even for circuits of logarithmic depth.

Inspired by these findings, a natural question to ask is whether we can \emph{localize} any given cost function, and thus avoid cost-function-dependent barren plateaus in the corresponding optimization landscape. 
For the problems considered in Refs.~\cite{Khatri2019,laroseVQSD,prietoVQLS}, the local cost function corresponding to the problem could be defined essentially by inspection of the global cost function. 
However, this is not necessarily possible in general and motivates us to open the discussion on whether perturbative gadgets could help extend results on the non-existence of barren plateaus for local cost functions to non-local ones.

\subsection{A reduction from \texorpdfstring{$k$}{k}-body to three-body cost functions}

Applying the above-presented gadget to substitute cost functions of the form in \eqref{eq:cost-observable}, allows us to use Theorem~\ref{theorem:main-result} and Corollary~\ref{corollary:minimas} to conclude the following about VQAs:
Given a cost function as in Eq.~\eqref{eq:cost-observable}, we can replace the $k$-body Hamiltonian $H^{\text{target}}$ with the three-body gadget Hamiltonian $H^{\text{gad}}$ from Definition~\ref{def:hamiltonians}, such that finding the minimum of the gadget Hamiltonian will bring us close to the low-energy subspace of the target Hamiltonian, and thus the minimum of the original cost function. 

Furthermore, by \emph{localizing} the original cost function, we obtain cost functions with provably non-vanishing gradients as shown in Refs.~\cite{uvarovBarrenPlateausCost2021a,cerezoCostFunctionDependent2021a}. Specifically, they prove that the corresponding gradients decrease at most polynomially in the number of qubits for the same ansatz of a logarithmic depth.
Thus, via Ref.~\cite[Theorem~2]{uvarovBarrenPlateausCost2021a}, we obtain the following result:

\begin{corollary}[Bound to the variance of the gradient of the cost function] \label{corollary:uvarov}
Consider the local cost function $C^{\text{gad}}(\boldsymbol{\theta})=\Tr[H^{\text{gad}}U(\boldsymbol{\theta})\rho U(\boldsymbol{\theta})^{\dagger}]$ defined by the gadget Hamiltonian in Definition~\ref{def:hamiltonians}, where $U(\boldsymbol{\theta})$ is the unitary corresponding to an alternating layered ansatz with $L$ layers, with each block of the ansatz drawn independently from a local unitary 2-design. 
Then, the variance of the gradient of the cost function is bounded from below by
\begin{equation}
    \text{Var}\!\left[\frac{\partial C^\text{gad}(\boldsymbol{\theta})}{\partial \boldsymbol{\theta}_\nu} \right] \geq \Omega \left( \mathrm{poly} \left( 3^{-L} \right)\right),
\end{equation}
where $\nu$ denotes the subset of parameters the gradient is evaluated for. We are thus guaranteed non-vanishing gradients for logarithmic circuit depths, i.e., for $L=\mathcal{O}(\log n)$.
\end{corollary}

In Fig.~\ref{fig:variances}, we present a proof-of-principle numerical demonstration of non-vanishing gradients.
At this point, we would like to point out that simply having non-vanishing gradients does not imply successful optimization, i.e., trainability.
Although the converse is true, one should, in general, be careful when implying trainability from the existence of gradients. 
Take, for example, the Hamiltonian $H = H^\text{global} \otimes \II + \II^{\otimes (n-1)} \otimes Z$. The vanishing gradients of the first, global term would effectively result in only the second term contributing to the gradients.
That is, if $H^\text{global}$ is a prototypical global observable afflicted by barren plateaus, it will barely contribute to the magnitude of the gradients. 
Consequently, the circuit will effectively be trained on the other, local, term $\II^{\otimes (n-1)} \otimes Z$, which has measurable gradients. 
In the end, one will successfully implement the gradient descent algorithm to a minimum, but only of the local part, and not of the full Hamiltonian.
In other words, the total cost function has measurable gradients, but these non-vanishing gradients do not help with the actual optimization problem.

Our gadget has gradients that do not vanish exponentially and a ground state that is close to the target solution. Unfortunately, it is not enough to fully avoid the core issue behind barren plateaus: the exponential increase in the number of required measurements. 
For the extreme case of $k=n$, the magnitude of the contribution of the target Hamiltonian in the effective Hamiltonian, or in other words the degeneracy splitting mentioned in Fig.~\ref{fig:spectrum_embedding}, is suppressed exponentially by $\lambda^n$. For the results from perturbation theory to hold, $\lambda$ has to be chosen to be small (see Theorem~\ref{theorem:main-result}), further amplifying this problem.
As a result, when reaching the low-energy subspace of the gadget Hamiltonian and therefore entering the regime described by the effective Hamiltonian, one would in practice need an exponentially increasing number of measurements to resolve the nuances in the cost function and find the global minimum.
While we can guarantee a vanishing error in the limit of $\lambda \rightarrow 0$ and that all terms of the gadget Hamiltonian contribute to the gradient, the main contributions to the gradients would be essentially those corresponding to reaching the unperturbed ground space, i.e., some state vector $\ket{\psi}=\ket{\phi}\otimes\ket{0}^{\otimes rk}$.
Therefore, actually finding the ground state of the gadget Hamiltonian might be prohibitively hard. 
However, this hardness can be present even for cost functions without barren plateaus, as discussed in detail in Refs.~\cite{anschuetzBarrenPlateausQuantum2022,bittelTrainingVariationalQuantum2021}.
Moreover, for numerical proof-of-principle demonstrations, we find that leaving the strict regime of perturbation strengths corresponding to theoretical guarantees can help the optimization performance as shown in Appendix~\ref{A:numerics}. 

For large $\lambda$ we do not have such strong analytical guarantees but allow for larger contributions with respect to the perturbative terms. 
Our training simulations reflect these insights as we obtain different behaviors for different perturbation strengths.
These observations are aligned with previous works on perturbative gadgets that have also found it beneficial to increase the perturbation strength, including increasing it to values outside the regime of convergence of perturbation theory~\cite{bravyiQuantumSimulationManyBody2008,bauschGadgets2020}.
Practical implementations might thus benefit highly from treating the perturbation strength $\lambda$ as a hyperparameter during optimization.

It is worth noting that while using the proposed gadget Hamiltonian requires $rk$ additional qubits and comprises $2rk$ terms instead of $r$, there exists a simple, optimal measurement scheme for estimating its expectation value, relying on only four different measurement-basis settings, inspired by Refs.~\cite{elbenRandomizedMeasurementToolbox2022,huangPredictingManyProperties2020f}.
First, by measuring all auxiliary qubits in the Pauli-$Z$ basis, we are able to estimate $H^\textrm{aux}$.
Then, consecutively measuring all auxiliary qubits in the Pauli-$X$ basis and simultaneously all target qubits in one of the three Pauli bases, we are able to estimate all terms $V_s$.

\begin{figure}
    \centering
    \includegraphics{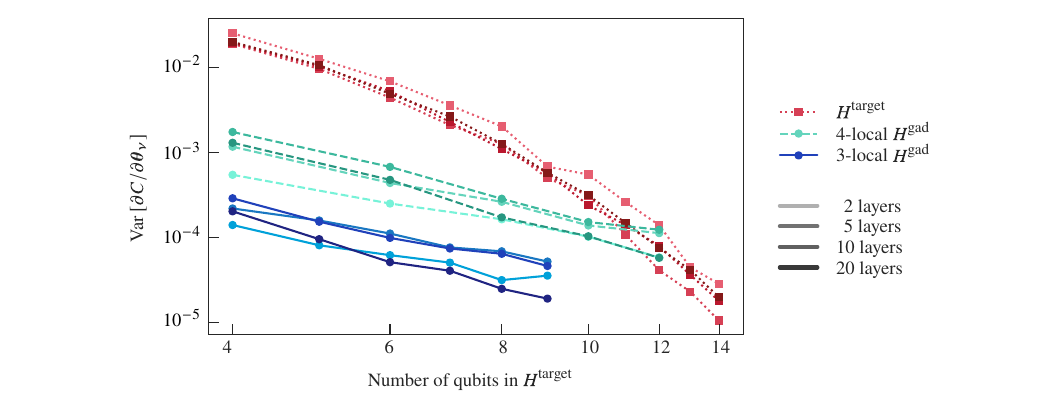}
    \caption{Demonstration of the theoretical guarantee from Corollary~\ref{corollary:uvarov}. We take as our global cost function the one in Eq.~\eqref{eq:cost-observable}, defined by the target Hamiltonian $H^\text{target} = \bigotimes_{i=1}^n Z_i$, and we use an alternating layered ansatz with different circuit depths. Shown is the variance of the gradient of this global cost function (dark pink), the local cost function based on our three-body gadget Hamiltonian from Definition~\ref{def:hamiltonians} (blue), and a similar local cost function defined by a four-body gadget Hamiltonian (turquoise) presented in Appendix~\ref{A:extensions}.}
    \label{fig:variances}
\end{figure}

\section{Inapplicability of the gadget by Jordan and Farhi for variational algorithms}\label{A:Farhi-inapplicability}

The perturbative gadget proposed by Jordan and Farhi~\cite{jordanPerturbativeGadgetsArbitrary2008a} relies strongly on the restriction of the allowed Hilbert space in which the state of the system can evolve. To be specific, in the language of Section~\ref{sec:recipe}, their gadget has a penalization Hamiltonian, as in Eq.~\eqref{eq-penalization_H}, of the form $H=\sum_{1\leq i_1<i_2\leq k} h_{i_1,i_2}$, with 
\begin{equation}\label{Aeq:JF-gadget}
    h_{i_1,i_2} = \frac{1}{2} (\II - Z_{i_1} \otimes Z_{i_2}),
\end{equation}
and the perturbation operator $A$ in Eq.~\eqref{eq-perturbation_A} is such that
\begin{equation}
    a_j  = X_j,\quad j\in\{1,2,\dotsc,k\}.
\end{equation}
However, since the ground space of their penalization Hamiltonian is two-dimensional and given by $\mathrm{span} \left\{ |0\rangle^{\otimes k},\ |1\rangle^{\otimes k}\right\}$, they have to enforce that the state of the system remains in the $+1$ eigenspace of the $A=X^{\otimes k}$ operator, thereby reducing the system to a one-dimensional ground state vector $| GS \rangle = \ket{GHZ} = ({1}/{\sqrt{2}}) \left( |0\rangle^{\otimes k} + |1\rangle^{\otimes k}\right)$. 
With this additional restriction, which can be fulfilled naturally in adiabatic quantum computing by initializing the auxiliary registers in GHZ states, it is easy to check that $X^{\otimes k}$ fulfills all properties stated in Eq.~\eqref{Aeq:conditions_A}. 
Unfortunately, this restriction is not possible outside of the adiabatic regime, and their gadget construction is consequently not applicable to our problem at hand.

However, let us now consider what happens when one cannot use initialization and adiabatic restriction to ensure a reduction of the reachable Hilbert space.
The unperturbed ground space of their penalization Hamiltonian $H$ in Eq.~\eqref{Aeq:JF-gadget} is 
\begin{equation} 
    \mathcal{E}^{(0)} = \operatorname{span} \left\{ |\varphi\rangle^\text{target} \bigotimes_{s=1}^r | \phi \rangle^\text{aux}_s : 
        |\varphi\rangle \in \mathcal{H}^\text{target},\,
        | \phi \rangle \in \operatorname{span} \left\{ |0\rangle^{\otimes k},\ |1\rangle^{\otimes k} \right\} \right\}
\end{equation}
and the corresponding projector is
\begin{equation}
    P_0 = \II \otimes \left( |0\rangle \! \<0|^{\otimes k} + |1\rangle \! \<1|^{\otimes k} \right)^{\otimes r}
    = \II \otimes \Pi_0^{\otimes r}.
\end{equation}
The fact that the dimension of $\mathcal{E}^{(0)}$ is not the same as the dimension of the target Hilbert space makes a considerable difference in Eq.~\eqref{a-eq:oder-k}, where the shifted contributions become
\begin{equation}
    \frac{\lambda^k}{\Xi} P_0 \left( \sum\limits_{s=1}^r c_s h_s \otimes X^{\otimes k}_s \right) P_0.
\end{equation}
Using the fact that 
\begin{equation}
    \Pi_0 = |GHZ_+ \rangle \! \< GHZ_+| + |GHZ_- \rangle \! \< GHZ_-|
\end{equation}
with
\begin{equation}
    |GHZ_{\pm} \rangle = \frac{1}{\sqrt{2}} \left( |0\rangle^{\otimes k} \pm |1\rangle^{\otimes k}\right),
\end{equation}
results in 
\begin{equation}
    \smash[b]{\frac{\lambda^k}{\Xi} \sum\limits_{s=1}^r} c_s h_s \otimes 
        \Pi_0^{\otimes s-1}
        \otimes (|GHZ_+ \rangle \! \< GHZ_+| - |GHZ_- \rangle \! \< GHZ_-|)
        \otimes \Pi_0^{\otimes r-s}.
\end{equation}
In other words, the effective Hamiltonian on the low-energy subspace of the gadget Hamiltonian is composed of a mixture of positive and negative contributions from each term as $\pm h_s$ is associated with different projectors on the different auxiliary registers. 
This prevents us from recovering the desired tensor-product structure between target and auxiliary registers with $H^\text{target}$ required to use this construction in practice, for our purposes; see also Ref.~\cite[Definition~11]{bauschGadgets2020}.

\section[Numerical simulations]{Numerical simulations of the application of the gadget for variational quantum algorithms}
\label{A:numerics}

To show how one can use the proposed gadget construction, we have studied its performance on the toy model $H^\text{target} =Z_1Z_2\dots Z_n$.
This specific example has already been used in the study of cost-function-dependent barren plateaus~\cite{holmesConnectingAnsatzExpressibility2022}. 
Furthermore, a single global Pauli string is the simplest example on which we can apply our gadget and that needs the least amount of resources. 
Indeed, we require $2n$ qubits to generate the corresponding cost function, while any additional global Hamiltonian term would imply $n$ additional qubits.
With this choice, having in mind the scaling of classical simulations, we can show the functioning of the gadget for the largest possible number of target qubits.
The purpose of our simulation is to demonstrate Theorem~\ref{theorem:main-result} and Corollary~\ref{corollary:minimas}; therefore, we performed exact, state-vector simulations, which are not subject to the considerations of exponential circuit evaluations presented previously.
We used a layered ansatz as done in Refs.~\cite{mccleanBarrenPlateausQuantum2018a, holmesConnectingAnsatzExpressibility2022}. While this ansatz does not fulfill the exact conditions of Corollary~\ref{corollary:uvarov}, it allows a comparison with previous works on the topic.

\subsection{Methods for gradient variance computations}
\label{A:variances-methods}

\begin{figure}
    \centering
    \includegraphics[width=.99\linewidth]{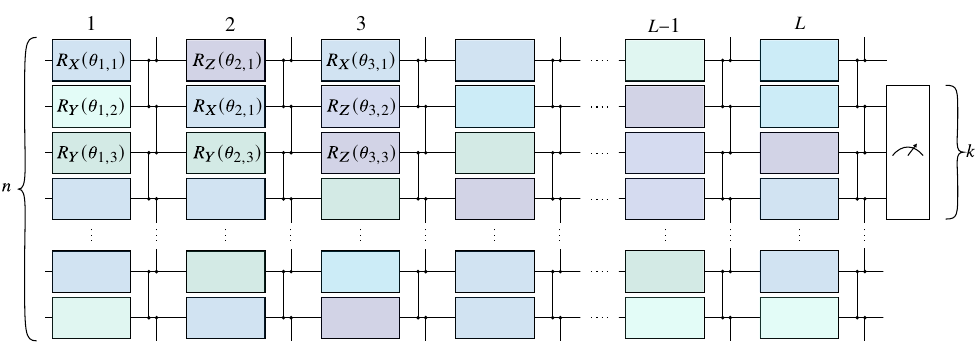}
    \caption{Circuit ansatz used for the PQC implemented in the simulations of the use of our perturbative gadget for VQE. Each layer is composed of randomly-chosen Pauli rotations on each qubit followed by controlled Pauli-$Z$ gates in a chain pattern.}
\end{figure}

As shown in Fig.~\ref{fig:variances}, the variances of the gradients of the energy for the target Hamiltonian decay exponentially, while those of the gadget Hamiltonian seem to decay subexponentially. 
The exponential scaling is more obvious when using linear-log axes, but the current log-log choice puts emphasis on the polynomial scaling of the curves for the gadget cost functions. Indeed, with this choice of axes, a polynomial scaling results in a straight line which seems to fit the data points well.
Although the presented regime does not show a clear advantage, the visible trend hints towards a crossover at larger qubit numbers that are out of reach for our classical simulator.

For these results, we have used the parameter on the $n^{\text{th}}$ qubit of the first layer when evaluating the variances. 
For the target Hamiltonian, this corresponds to the last qubit it acts upon, while it corresponds to the last qubit of the target register for the gadget Hamiltonian.
Based on the asymmetry of the gadget Hamiltonian with respect to the target and auxiliary registers, choosing the parameter related to the last qubit of the target register within the first layer guarantees a contribution of all Hamiltonian terms at any depth. 
However, we stress that the found trend is expected to be similar for all parameter choices, as indicated by Corollary~\ref{corollary:uvarov}.

\subsection{Training simulations}
\label{A:trainings}

\begin{figure*}
    \centering
    \includegraphics[width=\textwidth]{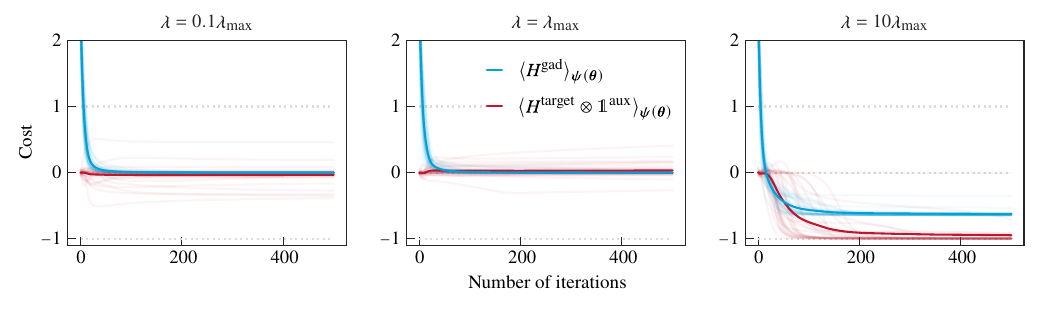}
    \caption{Training simulations using the gadget construction $H^\text{gad}$ introduced in Definition~\ref{def:hamiltonians} as cost to optimize the global Hamiltonian $H^\text{target} = Z^{\otimes 5}$. Each of the {30} runs is plotted in a light color while the more dark lines are the average over all runs.
    We employ an alternating layered ansatz on ten layers that is trained on $H^\text{gad}$ (blue).
    Furthermore, we evaluate the energy of $H^\text{target}\otimes\II$ (red) by measuring the expectation value of the target Hamiltonian for the output state of the circuit trained on the gadget Hamiltonian, hence tracing out the auxiliary part. While the theoretical guarantees of Theorem~\ref{theorem:main-result} only hold for $\lambda\leq\lambda_\text{max}$, this bound is not tight and we observe improved training for larger values. We thus propose treating $\lambda$ as an additional hyperparameter.}
    \label{fig:trainings}
\end{figure*}

Additionally to the gradient variance simulations, we have simulated gradient-based optimization of $H^\text{gad}$ while monitoring the expectation value of $H^\text{target}$ for different values of $\lambda$ as shown in Fig.~\ref{fig:trainings}.
From Theorem~\ref{theorem:main-result}, we know that when successfully reaching the global minimum of the expectation value of the gadget Hamiltonian, we can expect the output state to converge to a state close to the ground state of $H^\text{target}$. 
It has to be pointed out, however, that the subtleties of the low-energy part of the spectrum are exponentially suppressed in $\lambda$, which can result in slow optimizations or even exponential requirements in the presence of noise.
This can be seen in our experiments as the optimization stagnates for perturbation strengths within the range of validity of our theoretical results.
Nevertheless, we have successfully obtained minimizations of the expectation value of the target Hamiltonian for large values of $\lambda > \lambda_{\text{max}}$ within a reasonable number of iterations.
We thus propose using the interaction strength $\lambda$ as an additional hyperparameter and starting with larger values and reducing it later on for improved accuracy and faster training.
Furthermore, we note that the qubit ordering has an impact on the practical performance and that a different ordering has been employed for the variances and training plots. Further details are discussed in Appendix~\ref{A:reordering}.

Note that these training simulations are to be seen as proof-of-principle demonstrations, since, in the simulated regime, even the original, global Hamiltonian can be optimized.
Also, since the proposed gadget requires several times the number of qubits that the global Hamiltonian would, we quickly escape the realm of system sizes that can be efficiently simulated classically. 
Still, the scalings of our gadget construction remain 
compatible with the NISQ regime and could be a tool added to the arsenal of techniques for optimizing cost functions on near-term devices that are otherwise plagued by the barren plateau phenomenon.

\subsection{Performance improvements through qubit reordering}
\label{A:reordering}

\begin{figure}
    \centering
    \includegraphics[width=\textwidth]{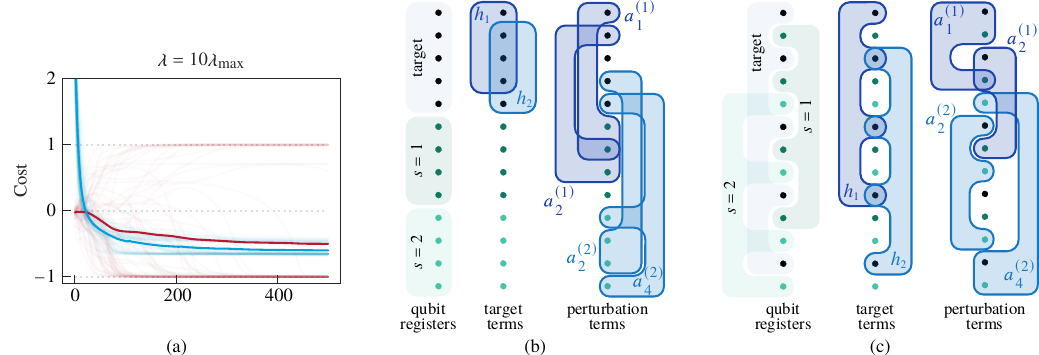}
    \caption{Two possible qubit orderings and the resulting simulations for the case of $n=5$, $k=4$, and $r=2$ (as in Fig.~\ref{fig:registers}). Panel (a) is the corresponding simulation to Fig.~\ref{fig:trainings} but using the basic ordering. Panels (b) and (c) show said ordering (used in Fig.~\ref{fig:variances}) and the improved one (used in Fig.~\ref{fig:trainings}) respectively. For each, first is the distribution of the 13 qubits into the three registers (left), then the terms from the target Hamiltonian and the qubits they act on (center), and finally a selection of the terms composing the gadget Hamiltonian (right). In contrast with Fig.~\ref{fig:registers}, we abandoned the circular layout and use a linear one to match the usual representation of the qubit register in a quantum computing circuit.}
    \label{fig:reordering}
\end{figure}

Until this point, we focused more on the conceptual message of using perturbative gadgets in VQAs to overcome cost-function-dependent barren plateaus and less on practical details relevant to their actual implementation. 
One of them is the qubit ordering in the joint register of the quantum computer or simulator. 
While our main theorems do not depend on the geometry of the quantum circuit, the couplings in the constructed gadget Hamiltonian do have a particular structure. 
Considering the state generated at the end of the PQC and the nearest-neighbor coupling of each layer of the used ansatz in combination with shallow circuit depths, it is reasonable to assume that the ordering of the qubits in a one-dimensional, joint register will have an effect.
Qubits coupled by the target Hamiltonian should probably be more entangled than those that are not coupled. 
We have thus explored different layouts of the target and auxiliary qubit registers with the goal of keeping coupled qubits closer.

First, let us discuss the straightforward qubit ordering used for the simulation results in Fig.~\ref{fig:variances}. As shown in Fig.~\ref{fig:reordering}(b), the qubits have been added to the total register as they appear in the gadget construction. That is, first the target register and then one auxiliary register after the other. 
The resulting coupling terms $\sigma_{s,j}^\text{target} \otimes X_{s,j}^\text{aux} \otimes X_{s,(j+1)\text{ mod }k}^\text{aux}$ are then divided with the target qubit being far from the two corresponding auxiliary qubits.
Training simulations in this setting (Fig.~\ref{fig:reordering} (a) ) were mostly successful, but in some instances ended up maximizing the target cost function instead of minimizing it.

Since the coupling graph resulting from $H^\text{gad}$ does not allow for a one-dimensional distribution such that only nearest-neighbor qubits are coupled, we settled for putting at least one of the relevant auxiliary qubits close to the target qubit.
To do so, we interleave the registers such that the first auxiliary qubit in $\sigma_{s,j}^\text{target} \otimes X_{s,j}^\text{aux} \otimes X_{s,(j+1)\text{ mod }k}^\text{aux}$ is added between the corresponding target qubit and the next target qubit. 
Due to the cyclic nature of the gadget construction, the second auxiliary qubit will be placed in the neighborhood of the following target qubit.
A simple example of this is presented in Fig.~\ref{fig:reordering} (c).
With this new ordering, we obtained the results shown in Fig.~\ref{fig:trainings}.
The comparison of the convergence properties of the training simulations in both settings provides numerical evidence of the role of qubit ordering when it comes to practical implementations of the introduced methods. 
First, the average convergence speed visibly improves when using the new ordering. 
Second, the better ordering does not exhibit any runs that maximize the target Hamiltonian as found for the first one, which could be attributed to the local minima of the gadget cost function.
Although this reordering does not pretend to be optimal but rather a first proposal, we could show for our example that considering the order of the qubits in combination with the structure of the minimized cost function can lead to clear improvements in the performance of VQAs in practical settings.

\section{A gadget for reducing measurement bases in the Hamiltonian readout problem}
\label{A:measurement_gadget}

The gadget discussed so far is a gadget tuned for maximally reducing the locality of the gadget Hamiltonian. 
There, we have also noted, that the energy of such a gadget construction can be estimated using only four measurement settings. 
This motivates the construction of a different gadget minimizing the resource overheads. 
Such a gadget is presented here.
Furthermore, since such a gadget can be expressed in four layers of mutually commuting terms, it could also be applicable to the fault-tolerant setting, where layers of mutually commuting $T$-gates can be executed in parallel by increasing the number of qubits~\cite{litinskiGameSurfaceCodes2019}.

\subsection{Gadget definition}
The gadget itself is constructed as follows: Let $H^{\text{target}}$ be an arbitrary Hamiltonian represented using $r$ differently weighted tensor products of Pauli operators, such that
\begin{equation}
    \label{eq:h_into_hxhyhz}
    H^{\text{target}}=\sum_{s=1}^rc_sH_s=\sum_{s=1}^rc_sH_s^XH_s^YH_s^Z,
\end{equation}
where each individual Pauli-word is decomposed into its Pauli-$X$, Pauli-$Y$, and Pauli-$Z$ part, e.g.,
\begin{align}
H_1 = & X_1 Y_2 X_3 Z_4,\\
H_1^X = & X_1 X_3,\\
H_1^Y = & Y_2 ,\\
H_1^Z = & Z_4.
\end{align}
Then, we choose our gadget Hamiltonian to be
\begin{equation}
    H^\text{gad}=H^\text{aux}+\lambda V
\end{equation}
with
\begin{equation}
    H^\text{aux}=-\sum_{i=1}^q Z_i
\end{equation}
acting as the unperturbed Hamiltonian on the $q$ auxiliary qubits and
\begin{equation}
    \label{eq:measurement_perturbation}
    \lambda V=\lambda\sum_{s=1}^r\sum_{A\in\{X,Y,Z\}}\tilde{c}_sH_s^A\otimes V_{s,A}^\text{aux}
\end{equation}
its perturbation. 
Since the goal is that products of $V_{s,A}^\text{aux}$ only become the identity if the corresponding product of $H_s^A$ is one of the target terms, we introduce
\begin{equation}
    \tau_i^X=X_{i-1}\qquad \tau_i^Y=X_{i} \qquad \tau_i^Z=X_{i-1}X_i
\end{equation}
for $X$ being the Pauli-$X$ operators acting on the auxiliary qubits.

The idea is to construct the terms $V_{s,A}^\text{aux}$ as products of some $\tau_i^A$ such that their product only results in the identity if we have $V_s^X V_s^Y V_s^Z$.
To do so, we define $M$ to be the set of all subsets of $\{1,2,\dots,q/2\}$ with an odd number of entries sorted with respect to their cardinality. There are
\begin{equation}
    \sum_{k \text{ odd}}^q \binom{q}{k} = 2^{q-1}
\end{equation}
many such sets. 
Then, constructing the perturbations $V_{s,A}^\text{aux}$ for the term $H_s^A$ using the set $M_s$, we define
\begin{equation}
\label{eq:nr_terms_qubits}
    V_{s,A}^\text{aux}=\prod_{i\in M_s}\tau_{2i}^A
\end{equation}
and
\begin{equation}
    \tilde{c}_s=\sqrt[3]{c_s|M_s|^2}.
\end{equation}
The ground state of $H^\text{aux}$ is given by $\ket{0}^{\otimes k}$.
Since the action of the perturbation is not immediately apparent, it is worthwhile to inspect it more closely. Imagine a system with $q=10$ auxiliary qubits $a_1,\dots,a_{q}$.
There, the individual $V_{i,A}^\text{aux}$ are chosen in the fashion of table \ref{tab:v_assignments}:
\begin{table}[ht]
    \centering
    \begin{tabular}{ r l || c c | c c | c c | c c | c c }
        &&$a_1$&$a_2$&$a_3$&$a_4$&$a_5$&$a_6$&$a_7$&$a_8$&$a_9$&$a_{10}$ \\ \hline
        $V_{1,X}^\text{aux}$ & $= \tau_{2}^X$ & $X$ & $\II$ &  &  &  &  &  &  &  &   \\
        $V_{1,Y}^\text{aux}$ & $= \tau_{2}^Y$ & $\II$ & $X$ &  &  &  &  &  &  &  &   \\
        $V_{1,Z}^\text{aux}$ & $= \tau_{2}^Z$ & $X$ & $X$ &  &  &  &  &  &  &  &   \\
        $V_{2,X}^\text{aux}$ & $= \tau_{4}^X$  &  &  & $X$ & $\II$  & &  &  &  &  &  \\
        $V_{3,X}^\text{aux}$ & $= \tau_{6}^X$  &  &  &  &  & $X$ & $\II$ &  &  &  &  \\
        $\vdots$ &  &  &  &  &  &  &  & $\ddots$ &  &  &  \\
        $V_{6,X}^\text{aux}$ & $= \tau_{2}^X \tau_{4}^X \tau_{6}^X$ & $X$ & $\II$ & $X$ & $\II$ & $X$ & $\II$ &  &  &  &   \\
        $V_{7,X}^\text{aux}$ & $= \tau_{2}^X \tau_{4}^X \tau_{8}^X$ & $X$& $\II$ & $X$ & $\II$ & & & $X$ & $\II$ & &    \\
        $\vdots$ &  &  &  &  &  &  &  &  &  & $\ddots$ &  \\
        $V_{16,X}^\text{aux}$ & $=\prod_{i=1}^5 \tau_{2i}^X $ & $X$ & $\II$ & $X$ & $\II$ & $X$ & $\II$ & $X$ & $\II$ & $X$ & $\II$
    \end{tabular}
    \caption{Exemplary assignments for the perturbation acting on ten auxiliary qubits for a computational Hamiltonian comprising 16 terms.}
    \label{tab:v_assignments}
\end{table}

In this fashion, we have that $V_s^X V_{s'}^Y V_{s''}^Z = \II$ only if $s = s' = s''$. 
In any other case, due to the perturbative expansion, the product does not contribute to the effective Hamiltonian.
Since there are two auxiliary qubits required as a unit, we obtain a qubit overhead of $q\geq2(\log_2(r)+1)$.
Here, we note that this perturbation structure is only one example and that more compact versions requiring fewer qubits might exist.

\subsection{Mathematical analysis}

The perturbation of $H^\text{aux}$ by $\lambda V$ results in a global energy shift and a splitting of the degeneracy of the ground space, mimicking $H^\text{target}$ and leading to the following result:

\begin{theorem}[Mathematical analysis]
\label{theo:measurement_gadget_effective}
    Given $H^\text{gad}$ acting on $n$ target and $q$ auxiliary qubits as defined above and $\lambda<1/2\norm{V}$ with the operator norm $\norm{...}$, $H_\text{eff}(H^\text{gad},d)$ will mimic $H^\text{target}$ at third order in perturbation theory with respect to $\lambda$, specifically
    \begin{equation}
        \label{eq:measurement_effective_gadget_hamiltonian}
        H_\text{eff}(H^\text{gad},2^n)=\Delta\Pi+\lambda^3H^\text{target}\otimes P_+ +\mathcal{O}(\lambda^4),
    \end{equation}
    where $\Delta$ is a constant energy shift and    $P_+=\proj{0}^{\otimes q}$ denotes the operator projecting all auxiliary qubits onto their ground state and $\Pi$ the projector on the support of $H_\text{eff}(H^\text{gad},2^k)$.
\end{theorem}

\begin{proof}
    The fundamentals of the perturbative expansion are presented above in Appendix~\ref{A:Bloch-expansion}, identical to those also used in Appendix~\ref{A:proof}. Consequently, the expansion converges if
    \begin{equation}
        \label{eq:measurement_convergence_criterium}
        \norm{\lambda V}<\frac{\gamma}{4},
    \end{equation}
    with $\gamma$ being equal to the smallest energy gap of $H^\text{aux}$, i.e., $\gamma=2$ for $H^\text{gad}$ leading to the condition that
    \begin{equation}
        \lambda\leq\frac{1}{2\norm{V}}.
    \end{equation}

    Let us now start by identifying nonzero terms in every order of perturbation theory. 
    Since (see Eq.~\eqref{Aeq:expansion-A}) all terms are sandwiched by $P_0=\II\otimes P_+$ operators projecting the system onto the ground space $\mathcal{E}^{(0)}$ of $H^\text{aux}$, all terms vanish that do not take a state from $\mathcal{E}^{(0)}$ and return it to $\mathcal{E}^{(0)}$.
    For the perturbation introduced in Equation~\eqref{eq:measurement_perturbation}, this is equivalent to saying that no combination of terms of $V$ resulting in an auxiliary space unequal to $\ket{0}^{\otimes q}$ will survive.
    Consequently, there will be no surviving terms in the first order of perturbation theory.

    The terms in second order originate from $V^2$ and belong to one of the following two categories: either they vanish because they excite the auxiliary space in two different ways not returning the state to $\mathcal{E}^{(0)}$ or they are squares of individual terms that survive since Pauli operators square to the identity. The latter are responsible for the constant energy shift $\Delta$.

    In analogy, all third-order terms vanish except for those of the form $V_s^X V_s^Y V_s^Z$ resulting in a degeneracy splitting caused by the application of $H_s^X H_s^Y H_s^Z$ on the original system, because only a combination of $\tau_s^X$, $\tau_s^Y$, and $\tau_s^Z$ will yield $\ket{0}^{\otimes q}\in\mathcal{E}^{(0)}$ on the auxiliary register after the perturbation.
    There are six possible permutations of these three perturbations, and the perturbation corresponding to the Pauli-$Z$ string can happen first, second or last, resulting in different energy penalties.
    Applying $V_i^Z$ first or last will result in one intermediate state with a penalty of $4\abs{M_s}$ and one with a penalty of $2\abs{M_s}$, whereas applying the $V_i^Z$ perturbation in the middle results in both intermediate states having a penalty of $2\abs{M_s}$. 
    Furthermore, from the sum in Eq.~\eqref{Aeq:expansion-A}, only the case $(\ell_1, \ell_2) = (1, 1)$ contributes, and not $(2, 0)$.
    Including the corresponding combinatorial factors we find
    \begin{align}
        H_\text{eff}(H^\text{gad},2^n)^{(3)} & 
        =\sum_{s=1}^r \frac{1}{\abs{M_s}^2} \Bigg( 2 \underbrace{\frac{1}{(-2)(-2)}}_{V_i^Z \text{ second}} + 4 \underbrace{\frac{1}{(-4)(-2)}}_{V_i^Z \text{ first or last}} \Bigg) \tilde{c}_s^3\lambda^3H_s^XH_s^YH_s^Z\otimes P_+\nonumber\\
        &=\lambda^3\sum_{s=1}^rc_sH_s\otimes P_+=\lambda^3H^\text{target}\otimes P_+.
    \end{align}
\end{proof}

The striking feature of these types of gadgets is, that their Hamiltonian can be grouped into four sums of qubitwise commuting Pauli operators, which can therefore be estimated using only four measurement bases:
\begin{equation}
\label{eq:grouping_terms}
    H^\text{gad} = \underbrace{-\sum_{i=1}^q \II\otimes Z_i}_{\II_\text{target}\otimes Z_\text{aux}} + \underbrace{\lambda\sum_{s=1}^r\tilde{c}_sH_s^X\otimes V_{s,X}^\text{aux}}_{X_\text{target}\otimes X_\text{aux}} + \underbrace{\lambda\sum_{s=1}^r\tilde{c}_sH_s^Y\otimes V_{s,Y}^\text{aux}}_{Y_\text{target}\otimes X_\text{aux}} + \underbrace{\lambda\sum_{s=1}^r\tilde{c}_sH_s^Z\otimes V_{s,Z}^\text{aux}}_{Z_\text{target}\otimes X_\text{aux}}.
\end{equation}

However, as mentioned above, the scaling in $\lambda$ is unfavorable except for settings in which it is excessively expensive to change measurement settings but cheap to evaluate a circuit given a chosen setting.

\end{document}